\documentclass[a4paper,12pt,oneside]{article}
\usepackage{amsmath,amssymb,mathtools,bm}
\usepackage[utf8]{inputenc}
\usepackage{csquotes}
\usepackage[english]{babel}
\usepackage[T1]{fontenc}
\usepackage{amstext}
\usepackage{amsthm}
\usepackage{bbm}
\usepackage{enumerate}
\usepackage{hyperref}
\usepackage[nameinlink,capitalise]{cleveref}
\usepackage{braket}
\usepackage{dsfont}
\usepackage{color}
 \usepackage{tikz}
\usepackage{pgfplots}

\usepackage{authblk}

\usepackage[verbose=true,letterpaper]{geometry}
\AtBeginDocument{
	\newgeometry{
		textheight=9.5in,
		textwidth=7in,
		top=1in,
		headheight=14.5pt,
		headsep=10pt,
		footskip=25pt
	}
}

\usepackage{fancyhdr}
\fancyheadoffset{0pt}
\usepackage{titlesec}
\titleformat{\section}{\bfseries\scshape\Large}{\thesection}{1em}{}{}
\titleformat*{\subsection}{\scshape\bfseries\large}

\renewenvironment{thebibliography}[1]{
	\begin{oldthebibliography}{#1}
		\footnotesize
		\setlength{\itemsep}{0em}
		\setlength{\parskip}{0em}
	}
	{
	\end{oldthebibliography}
}

\numberwithin{equation}{section}
\newtheorem{thm}{Theorem}[section]
\newtheorem{lem}[thm]{Lemma}
\newtheorem{prop}[thm]{Proposition}
\newtheorem{cor}[thm]{Corollary}

\theoremstyle{definition}

\renewcommand*{\thehyp}{\Alph{hyp}}

\theoremstyle{remark}
\newtheorem{rem}[thm]{Remark}

\crefname{hyp}{Hypothesis}{Hypotheses}
\Crefname{hyp}{Hypothesis}{Hypotheses}
\crefname{lem}{Lemma}{Lemmas}
\Crefname{lem}{Lemma}{Lemmas}
\crefname{thm}{Theorem}{Theorems}
\Crefname{thm}{Theorem}{Theorems}
\crefname{prop}{Proposition}{Propositions}
\Crefname{prop}{Proposition}{Propositions}
\crefname{enumi}{}{}
\Crefname{enumi}{}{}
\creflabelformat{enumi}{#2(#1)#3}
\crefname{equation}{}{}
\Crefname{equation}{}{}
\crefname{rem}{Remark}{Remarks}
\Crefname{rem}{Remark}{Remarks}

\makeatletter
\renewcommand{\@upn}{} 
\makeatother

\usepackage{etoolbox}
\makeatletter
\patchcmd{\endthm}{\@endpefalse}{}{}{}
\patchcmd{\endcor}{\@endpefalse}{}{}{}
\patchcmd{\endlem}{\@endpefalse}{}{}{}
\patchcmd{\endprop}{\@endpefalse}{}{}{}
\patchcmd{\endproof}{\@endpefalse}{}{}{}
\makeatother

\usepackage[inline]{enumitem}
\newlist{enumthm}{enumerate}{1} 
\setlist[enumthm]{label=\upshape(\roman*),ref=\thethm~(\roman*)}  
\crefalias{enumthmi}{thm} 
\newlist{enumcor}{enumerate}{1}
\setlist[enumcor]{label=\upshape(\roman*),ref=\thecor~(\roman*)}
\crefalias{enumcori}{cor}
\newlist{enumlem}{enumerate}{1}
\setlist[enumlem]{label=\upshape(\roman*),ref=\thelem~(\roman*)}
\crefalias{enumlemi}{lem}
\newlist{enumprop}{enumerate}{1}
\setlist[enumprop]{label=\upshape(\roman*),ref=\theprop~(\roman*)}
\crefalias{enumpropi}{prop}
\newlist{enumhyp}{enumerate}{1}
\setlist[enumhyp]{label=\upshape(\roman*),ref=\thehyp~(\roman*)}
\crefalias{enumhypi}{hyp}
\newlist{enumproof}{enumerate*}{1}
\setlist[enumproof]{label=\upshape(\roman*)}
\newlist{enumdef}{enumerate}{1}
\setlist[enumdef]{label=\upshape(\roman*),ref=\thedefn~(\roman*)}
\crefalias{enumdefi}{defn}
\newlist{enumrem}{enumerate}{1}
\setlist[enumrem]{label=\upshape(\roman*),ref=\thedefn~(\roman*)}
\crefalias{enumrem}{rem}

\makeatletter
\newcounter{subcreftmpcnt} %
\newcommand\romansubformat[1]{(\roman{#1})} 
\def\subcref{\@ifstar\@@subcref\@subcref}
\newcommand\@subcref[2][\romansubformat]{%
	\ifcsname r@#2@cref\endcsname
	\cref@getcounter {#2}{\mylabel}%
	\setcounter{subcreftmpcnt}{\mylabel}%
	\hyperref[#2]{\romansubformat{subcreftmpcnt}}%
	\else ?? \fi}   
\newcommand\@@subcref[2][\romansubformat]{%
	\ifcsname r@#2@cref\endcsname
	\cref@getcounter {#2}{\mylabel}%
	\setcounter{subcreftmpcnt}{\mylabel}%
	\romansubformat{subcreftmpcnt}%
	\else ?? \fi}   
\makeatother

\makeatletter
\DeclareRobustCommand{\crefnosort}[1]{%
	\begingroup\@cref@sortfalse\cref{#1}\endgroup
}
\makeatother

\newcommand{\cB}{{\mathcal B}}\newcommand{\cC}{{\mathcal C}}
\newcommand{\cD}{{\mathcal D}}\newcommand{\cF}{{\mathcal F}}
\newcommand{\cH}{{\mathcal H}}\newcommand{\cI}{{\mathcal I}}

\newcommand{\cN}{{\mathcal N}}

\newcommand{\fF}{{\mathfrak F}}

\newcommand{\fh}{{\mathfrak h}}

\newcommand{\fp}{{\mathfrak p}}
\newcommand{\fs}{{\mathfrak s}}

\newcommand{\BC}{{\mathbb C}}

\newcommand{\BN}{{\mathbb N}}
\newcommand{\BR}{{\mathbb R}}

\usepackage{dsfont} 

\newcommand{\dsone}{{\mathds 1}}

\usepackage{mathrsfs}
\newcommand{\sB}{{\mathscr B}}\newcommand{\sC}{{\mathscr C}}
\newcommand{\sD}{{\mathscr D}}

\newcommand{\sP}{{\mathscr P}}

\newcommand{\sW}{{\mathscr W}}

\newcommand{\sfC}{{\mathsf C}}

\newcommand{\sfH}{{\mathsf H}}

\newcommand{\sfc}{{\mathsf c}}
\newcommand{\sfd}{{\mathsf d}}\newcommand{\sff}{{\mathsf f}}

\newcommand{\sfm}{{\mathsf m}}\newcommand{\sfn}{{\mathsf n}}
\newcommand{\sfr}{{\mathsf r}}
\newcommand{\sfs}{{\mathsf s}}

\newcommand{\sfy}{{\mathsf y}}

\newcommand{\rme}{{\mathrm e}}

\newcommand{\IN}{\BN}\newcommand{\IR}{\BR}\newcommand{\IC}{\BC}
\newcommand{\N}{\BN}\newcommand{\R}{\BR}\newcommand{\C}{\BC}
\newcommand{\RR}{\BR}

\newcommand{\hs}{\fh}\newcommand{\HS}{\cH}

\newcommand{\eps}{\varepsilon}\newcommand{\ph}{\varphi}\newcommand{\vr}{\varrho}

\newcommand{\e}{\rme}\renewcommand{\i}{\sfi}\newcommand{\Id}{\dsone} \renewcommand{\d}{\sfd}

\renewcommand{\Re}{\operatorname{Re}}

\newcommand{\wh}[1]{\widehat{#1}}\newcommand{\wt}[1]{\widetilde{#1}}\renewcommand{\bar}[1]{\overline{#1}}

\DeclareFontFamily{U}{mathx}{\hyphenchar\font45}
\DeclareFontShape{U}{mathx}{m}{n}{
	<5> <6> <7> <8> <9> <10>
	<10.95> <12> <14.4> <17.28> <20.74> <24.88>
	mathx10
}{}
\DeclareSymbolFont{mathx}{U}{mathx}{m}{n}
\DeclareFontSubstitution{U}{mathx}{m}{n}
\DeclareMathAccent{\widecheck}{0}{mathx}{"71}
\DeclareMathAccent{\wideparen}{0}{mathx}{"75}

\DeclareFontFamily{OMX}{MnSymbolE}{}
\DeclareFontShape{OMX}{MnSymbolE}{m}{n}{
	<-6>  MnSymbolE5
	<6-7>  MnSymbolE6
	<7-8>  MnSymbolE7
	<8-9>  MnSymbolE8
	<9-10> MnSymbolE9
	<10-12> MnSymbolE10
	<12->   MnSymbolE12}{}
\DeclareSymbolFont{mnlargesymbols}{OMX}{MnSymbolE}{m}{n}
\SetSymbolFont{mnlargesymbols}{bold}{OMX}{MnSymbolE}{b}{n}
\DeclareMathDelimiter{\llangle}{\mathopen}{mnlargesymbols}{'164}{mnlargesymbols}{'164}
\DeclareMathDelimiter{\rrangle}{\mathclose}{mnlargesymbols}{'171}{mnlargesymbols}{'171}
\DeclareMathDelimiter{\lsem}{\mathopen}{mnlargesymbols}{'102}{mnlargesymbols}{'102}
\DeclareMathDelimiter{\rsem}{\mathclose}{mnlargesymbols}{'107}{mnlargesymbols}{'107}
\DeclareMathDelimiter{\langlebar}{\mathopen}{mnlargesymbols}{'152}{mnlargesymbols}{'152}
\DeclareMathDelimiter{\ranglebar}{\mathclose}{mnlargesymbols}{'157}{mnlargesymbols}{'157}
\DeclareMathDelimiter{\lWavy}{\mathopen}{mnlargesymbols}{'137}{mnlargesymbols}{'137}
\DeclareMathDelimiter{\rWavy}{\mathopen}{mnlargesymbols}{'137}{mnlargesymbols}{'137}

\newcommand{\FGamma}{\Gamma}
\newcommand{\FS}{\cF}\newcommand{\dG}{\sfd\FGamma}

\def\titlename{\scshape  Non-Fock Ground States in the Nelson Model}
\title{\LARGE\scshape  Non-Fock Ground States in the\\ Translation-Invariant Nelson Model\\  Revisited Non-Perturbatively}
\newcommand{\shortauthors}{D. Hasler, B. Hinrichs, O. Siebert}

\author{David Hasler\thanks{\texttt{david.hasler@uni-jena.de}}}
\author{Benjamin Hinrichs\thanks{\texttt{benjamin.hinrichs@math.upb.de}\\Present Address: Universit\"at Paderborn, Institut f\"ur Mathematik, Warburger Str. 100, D-33098 Paderborn, Germany}}
\affil{Friedrich-Schiller-Universität Jena\newline{\small Department of Mathematics\\  Ernst-Abbe-Platz 2\\  07743 Jena\\ Germany}\vspace*{.5em}}
\author{Oliver Siebert\thanks{\texttt{oliver.siebert@uni-tuebingen.de}}}
\affil{Eberhard-Karls-Universität Tübingen \newline{\small Institute of Mathematics\\Auf der Morgenstelle 10\\D-72076 Tübingen\\Germany}
}
\date{}

\renewcommand{\i}{{\mathsf{i}}}

\newcommand{\chr}[1]{\mathbf{1}_{#1}}

\newcommand{\mm}{{\mathbf k}}

\newcommand{\disp}{\Theta}
\newcommand{\ppn}{\psi_{P,n}}\newcommand{\fpn}{g_{P,n}}
\newcommand{\ppnk}{\psi_{P,n_k}}\newcommand{\ppi}{\psi_{P,\infty}}
\newcommand{\CT}[1]{\sfC_{\Theta,#1}}\newcommand{\Cw}{\sfC_\omega}\newcommand{\Cwn}{\sfC_{\omega_n}}
\newcommand{\Rnp}[1]{R_n(P,#1)}

\renewcommand{\sc}[1]{\left< #1 \right>}
\newcommand{\nn}[1]{\left\Vert #1 \right\Vert}
\newcommand{\abs}[1]{\left| #1 \right|}

\newcommand{\f}{{\sff}}
\newcommand{\nr}{{\sfn\sfr}}
\newcommand{\sr}{{\sfs\sfr}}
\newcommand{\Hf}{\dG(\omega)}

\newcommand{\Pf}{P_\f}
\newcommand{\Pfi}[1]{P_{\f,#1}}
\renewcommand{\:}{\colon}

\newcommand{\pos}{\sP\!\!\!_{-v}}

\newcommand{\DP}{{\Delta_{\disp,\omega,v,P}}}
\newcommand{\DPn}{{\Delta_{\disp,\omega_n,v,P}}}
\newcommand{\CP}{\sfC_P}

\newcommand{\mw}{\sfm_\omega}
\newcommand{\mwn}{\sfm_{\omega_n}}
 \newcommand{\dGw}{(\dG(\omega\wedge 1) +1)}

\newcommand{\vv}{\boldsymbol{v}}

\begin{document}
	
	\maketitle\thispagestyle{empty}\vspace*{-2em}
	\begin{abstract}\noindent
		The Nelson model, describing a quantum mechanical particle linearly coupled to a bosonic field, exhibits the infrared problem in the sense that no ground state exists at arbitrary total momentum. However, passing to a non-Fock representation, one can prove the existence of so-called dressed one-particle states. In this article, we give a simple non-perturbative proof for the existence of such one-particle states at arbitrary coupling strength and for almost all total momenta in a physically motivated momentum region. Our results hold both for the non- and the semi-relativistic Nelson model.	
	\end{abstract}

\section{Introduction}

In models describing the interaction of particles with a quantized radiation field, infrared divergences lead to severe technical issues in the rigorous construction of scattering theories.
The physical interpretation for this phenomenon, first described in an article by Bloch and Nordsieck \cite{BlochNordsiek.1937}, is that a cloud of infinitely many photons with small energies emerge in the interaction, which breaks the square-integrability of the system.
Mathematically, this reflects in the fact that Hamiltonians of such systems do not need to have a ground state. Since ground states of small total momenta of translation-invariant systems  are essential in the construction of scattering states, this poses a key technical difficulty induced by the infrared problem.

In this paper, we focus on the infrared problem in the translation-invariant
Nelson model, prominently treated by Edward Nelson \cite{Nelson.1964}, in which a charged particle is linearly coupled to a scalar bosonic field.
It has been an important testing ground for the mathematical descriptions of infrared critical scattering phenomena in the past decades.
In his PhD thesis, J\"urg Fr\"ohlich studied the infrared problem in the Nelson model \cite{Frohlich.1973,Frohlich.1974}.
Especially, in \cite{Frohlich.1974}, he proposed dressed one-electron states in an inequivalent representation of the Nelson model.
Using an iterative perturbative approach, Pizzo was able to explicitly construct such states in a Hilbert space by means of a dressing transformation and an iterative algorithm used as an infrared limiting procedure \cite{Pizzo.2003}. He was then able to construct asymptotic scattering states from these in \cite{Pizzo.2005}.
In recent years, Dybalski and Pizzo used this approach to study Coulomb scattering in Nelson's model in a series of papers \cite{DybalskiPizzo.2014,DybalskiPizzo.2018,DybalskiPizzo.2019,DybalskiPizzo.2022}.
It is also noteworthy, that a novel construction of scattering states without an infrared limiting procedure was developed in the recent article \cite{BeaudDybalskiGraf.2021}.
The study of such one-electron states has been extended to a variety of different models, due to its importance in scattering theory. This includes the ultraviolet renormalized Nelson model \cite{BachmannDeckertPizzo.2012}, the non-relativistic Pauli--Fierz model \cite{ChenFroehlichPizzo.2009,ChenFroehlichPizzo.2010,HaslerSiebert.2022} and the semi-relativistic Nelson \cite{DeckertPizzo.2014} as well as Pauli--Fierz models \cite{KoenenbergMatte.2014}. Furthermore, it was shown  that these dressed one-electron states  are indeed  ground states of 
suitably  infrared renormalized fiber Hamiltonians  \cite{BachmannDeckertPizzo.2012,KoenenbergMatte.2014,HaslerSiebert.2022}.
This  underlines the  picture  that the dressed one-electron states  live in  non-Fock  representations of the canonical commutation relations, cf.    \cite{ChenFroehlich.2007,Arai.2001}.

In this article, we give a non-perturbative proof for the existence of dressed one-electron states in a class of Nelson-type models, 
which  include  the standard non-relativistic Nelson model introduced in \cite{Nelson.1964} and the semi-relativistic Nelson model, cf. \cite{Gross.1973}.  Moreover, we show that these dressed one-electron states are ground states of infrared transformed
fiber Hamiltonians. 
Our method to construct the dressed one-electron states  uses a compactness argument tracing back to \cite{GriesemerLiebLoss.2001}, which 
was recently applied to the translation-invariant Pauli--Fierz model in \cite{HaslerSiebert.2020,HaslerSiebert.2022} as well as 
the spin boson model in \cite{HaslerHinrichsSiebert.2021a}. 
Instead of Rellich's criterion, we use  the more general Fr\'echet--Kolmogorov--Riesz theorem to establish compactness, cf.   \cite{Matte.2016,HiroshimaMatte.2019,Hinrichs.2022b} for previous treatments of compactness in Fock spaces via this approach. 
To make this article as selfcontained as possible and to provide a new simpler proof of Fock space compactness criteria, we derive an abstract compactness theorem in this article, which can be more directly verified and hence significantly simplifies our proof of compactness.
Our result holds for all values of the coupling constants and
for almost all physically relevant total momenta, i.e., momenta with absolute value smaller than one in the non-relativistic case and almost all total momenta in the semi-relativistic case, respectively. Moreover, combining our results with  regularity results of the ground state energy 
\cite{AbdessalamHasler.2012}, we obtain   one-electron states in the non-relativistic case    for any total momentum less than one
provided the coupling is  sufficiently small, as a corollary.
Thus, we complement  and  simplify the existence results obtained in \cite{Pizzo.2003,BachmannDeckertPizzo.2012},  which in addition we extend  
to semi-relativistic dispersion relations, also see \cite{DeckertPizzo.2014}.
For the construction of scattering states, e.g., in \cite{BeaudDybalskiGraf.2021}, it is important to have the one-electron states chosen continuous in the total momentum. This usually follows directly when employing the existing perturbative methods \cite{Pizzo.2003,DybalskiPizzo.2019}, but can not be inferred solely using a compactness argument. However, we demonstrate that continuity w.r.t. the total momentum follows, when additionally employing positivity and non-degeneracy of the ground states. To the authors knowledge, these are not available for the infrared transformed Hamiltonians at this point and are left for future research.

In this article, for technical convenience, we do not attempt to work with an ultraviolet renormalized model, see   \cite{Nelson.1964,Cannon.1971,Gross.1973,Sloan.1974}
and for recent work see, e.g., \cite{LampartSchmidt.2019,MatteMoller.2018,Schmidt.2019,HinrichsMatte.2022}. 
We  plan to address the extension of the methods used in the present paper  to the ultraviolet renormalized case in future work.

\subsection*{Structure and Notation}

This article is structured as follows:
\begin{itemize}
	\item[-] In \cref{sec:results}, we introduce the necessary notation and present our two main results on the non- and semi-relativistic Nelson model, cf. \cref{mainthm1,mainthm2}.
	\item[-] In \cref{sec:compact}, we derive a simple criterion for compactness of sets in Fock spaces, cf. \cref{comfock}.
	\item[-] In \cref{sec:transformed}, we derive the non-Fock representation of the non- and semi-relativistic Nelson Hamiltonians and prove the convergence result \cref{th:P2 diff}, which is the key ingredient to the proof for the existence of infrared-renormalized Hamiltonians in this representation.
	\item[-] In \cref{sec:conv}, we discuss ground states of Nelson-type models with general dispersion relations of the particle, their infrared properties and (most important) a compactness result for ground states and dressed ground states, cf. \cref{thm:compactness}. This can be considered an essential result of this article, since it implies the existence of ground states for the infrared-renormalized Hamiltonian in the non-Fock representation.
	\item[-] We collect well-known technical ingredients of our proofs in the appendices.
\end{itemize}
Let us, throughout this article, fix the following conventions:
\begin{itemize}
	\item[-] The dimension $d\in\IN$ is a fixed number throughout this article.
	\item[-] For $a,b\in\IR$, we write $a\wedge b = \min\{a,b\}$ and $a\vee b = \max\{a,b\}$.
	\item[-] Given $m\in\IN$, we write $B_r(x)  = \{y\in\IR^m:|x-y|<r\}$ for the open ball of radius $r>0$ around $x\in\IR^m$.
	\item[-]For any set $M\subset \IR^m$, we denote its complement by $M^\sfc = \IR^m\setminus M$.
	\item[-] For a Hilbert space operator $A$, we denote its domain as $\cD(A)$ and its spectrum as $\sigma(A)$.
	\item[-] We say a selfadjoint lower-semibounded operator $A$ has a {\em ground state} if $\inf\sigma(A)$ is an eigenvalue of $A$.
	\item[-] Given two selfadjoint operators $A$ and $B$, we write $A\le B$ if $\cD(A)\supset \cD(B)$ and $\braket{x,Ax}\le\braket{x,Bx}$ for all $x\in\cD(B)$.
\end{itemize}

\section{Main Results}\label{sec:results}
In this \lcnamecref{sec:results}, we present the main results of  this article concerning the non- and semi-relativistic Nelson model. First, in \cref{sec:res.Fock}, we introduce the necessary Fock space notation. It can be skipped by the experienced reader familiar with Fock space calculus. Then, in \cref{sec:res.model}, we define  generalized Nelson models,  which  cover  
the   two special cases of the non-relativistic and semi-relativistic Nelson model.
In the final  \cref{sec:res.res}, we present the main results in \cref{mainthm1,mainthm2}.

\subsection{Fock Space Notions}\label{sec:res.Fock}
Let us briefly introduce the Fock space notation needed for this article. For a more detailed introduction, we refer to the textbooks \cite{ReedSimon.1975,Parthasarathy.1992,BratteliRobinson.1996,Arai.2018}.
Well-known properties of Fock space operators are collected in \cref{sec:operator}.

Given a complex Hilbert space $\hs$, we define the bosonic Fock space over $\hs$ as
\begin{equation*}
	\FS(\hs) \coloneqq \IC\oplus \bigoplus_{n=1}^\infty \hs^{\otimes_\sfs n},
\end{equation*}
where $\otimes_\sfs$ denotes the symmetric tensor product, and as usually identify $(L^2(M))^{\otimes_\sfs} = 
L^2_{\sfs\sfy\sfm}(M^{ n})$ for any domain $M\subset \IR^d$, where in the latter we symmetrize over the $n$ variables from $\IR^d$. Especially, we write
\begin{equation*}
	\FS \coloneqq \FS(L^2(\IR^d)).
\end{equation*}
%
For selfadjoint operator $A$ in $\mathfrak{h}$, we define 
\begin{equation*}
	\dG (A) \coloneqq 0  \oplus \bigoplus_{n=1}^\infty \sum_{j=1}^n  \Id^{\otimes j-1} \otimes A  \otimes \Id^{\otimes n-j}  .
\end{equation*}
For $g \in L^2(M)$, we define the exponential  vector  $\epsilon(g)$ in $\mathcal{F}(L^2(M))$ by 
$(\epsilon(g))_0= 1$ and $(\epsilon(g))_n = (n!)^{-1/2}g^{\otimes_s n}$ for $n \in \N$. 
One can show that the linear span of such vectors is dense in $\mathcal{F}(L^2(M))$ \cite{Parthasarathy.1992}. 
For $f\in L^2(M)$, we   define the Weyl operator $W(f)$   as the unique unitary operator satisfying
\begin{equation}\label{def:Weyl}
	W(f)\epsilon(g) = \e^{-\|f\|^2/2-\braket{f,g}}\epsilon(f+g) \qquad \mbox{for all}\ f\in L^2(M).
\end{equation}
The map  $t\mapsto W(-\i t f)$ defines a strongly continuous unitary one parameter  group and we denote its selfadjoint generator as $\ph(f)$, the so-called field operator corresponding to $f$.

Given any $f\in L_{\sfs\sfy\sfm}^2((\IR^d)^n)$,  $n\in\IN$, 
by the Fubini--Tonelli theorem,
\begin{equation*}
	a_k f \coloneqq \sqrt{n}f(k,\cdots)\in L^2_{\sfs\sfy\sfm}((\IR^d)^{n-1})
\end{equation*}
is well-defined for almost every $k\in\IR^d$. Given $\psi\in\FS$, for almost all $k\in\IR^d$, we can hence define
\begin{equation}
	a_k\psi \coloneqq (a_k\psi^{(n)})_{n\in\IN} \in \IC \times \Big(\bigtimes_{n\in\IN} L^2_{\sfs\sfy\sfm}((\IR^d)^n)\Big).
\end{equation}
It is well-known (and for most readers probably the more familiar definition of the field operator) that, given any $f\in L^2(\IR^d)$, if $\psi\in\cD(\ph(f))$, we have $a_k\psi \in \FS $ for almost all $k\in \IR^d$ and
\[  \braket{\phi,\ph(f)\psi} = \int_{\IR^d} \big( \! \braket{f(k)\phi,a_k\psi} + \braket{a_k\phi,f(k)\psi} \! \big)\d k \qquad \mbox{for all}\ \phi,\psi\in\cD(\ph(f)).   \]


\subsection{General Nelson Models}\label{sec:res.model}

The class of models which  we investigate is determined by three functions: the dispersion relation of the particle $\disp$, the dispersion relation of the bosons $\omega$ and the so-called form factor describing their interaction $v$. Throughout this article, we will assume the following standing assumptions to be satisfied:
\begin{enumerate}[label = (H\arabic*),ref=H\arabic*]
	\item\label{hyp:Omega} $\disp\in \cC^\infty(\IR^d;[0,\infty])$ is analytic and there exist constants $\CT{1,1}$, $\CT{1,2}$, $\CT{2}\ge0$ such that
	\begin{enumerate}[label = -]
		\item $|\nabla\disp(p)| \le \CT{1,1} + \CT{1,2}\disp(p)$ for all $p\in\IR^d$,
		\item the Hessian $\sfH\disp(p)$ is uniformly bounded in $p\in\IR^d$, more precisely $\displaystyle \CT2 = \frac 12 \sup_{p\in\IR^d}\|\sfH\disp(p)\|_{\cB(\IR^d)}<\infty$.
	\end{enumerate}
	\item\label{hyp:omega} $\omega:\IR^d\to [0,\infty)$ is continuous and subadditive, i.e., 
	$\omega(k_1 + k_2) \leq \omega(k_1) + \omega(k_2)$ for all $k_1,k_2 \in \R^d$, 
	and strictly positive almost everywhere. 
	\item\label{hyp:div} $\lim_{|p|\to\infty}\disp(p) = \lim_{|k|\to\infty}\omega(k) = \infty$.
	\item\label{hyp:v} $v\in L^2(\IR^d)$ such that  $\omega^{-1/2}v\in L^2(\IR^d)$.
\end{enumerate}
For $j=1,\ldots,d$, let us define the multiplication operator $\mm_j$ acting on $L^2(\IR^d)$ by $(\mm_j f)(k)\coloneqq k_jf(k)$ and use the vector notation $\mm \coloneqq (\mm_1,\ldots,\mm_d)$. 
Then, we define the momentum of the quantum field as the vector of selfadjoint operators given by $\Pf \coloneqq \dG(\mm) \coloneqq (\dG(\mm_1),\ldots,\dG(\mm_d))$.

The translation-invariant Nelson-type Hamiltonians which  we investigate are  defined by 
\begin{align}\label{def:general}
	H(\disp,\omega,v,P) \coloneqq \disp(P-\Pf) + \Hf + \ph(v) \quad \mbox{for}\ P\in\RR^d.
\end{align}
One can show  by standard arguments that \cref{def:general} defines a selfadjoint lower-semibounded operator, given our standing assumptions. This is the content of the following lemma. 
\begin{lem} 
	The operator $H(\disp,\omega,v,P)$ given on the domain $\cD(\disp(\Pf))\cap \cD(\Hf)$ by \cref{def:general} is selfadjoint and lower-semibounded for any $P\in\IR^d$.
\end{lem}
\begin{proof}
	By \cref{lem:dispindP}, the operator $H(\disp,\omega,0,P)$ is selfadjoint and non-negative on the given domain as the sum of commuting selfadjoint non-negative operators. Further, by \cref{th:phi estimate}, $\ph(v)$ is infinitesimally bounded w.r.t. $H(\disp,\omega,0,P)$. Hence, the statement follows from the Kato--Rellich theorem.
\end{proof}
In this introduction, we will restrict to a discussion of our two main examples: the {\em non-relativistic Nelson model} given by the choice $\disp(p)=\disp_\nr(p)\coloneqq\frac 1{2M}p^2$ and the {\em semi-relativistic Nelson model} given by the choice $\disp(p)=\disp_{\sr}(p)\coloneqq\sqrt{M^2+p^2}$. In both cases, $M > 0$ is the mass of the particle, which we discuss as a parameter of the model. The corresponding Hamiltonians will be denoted by
\begin{align}\label{def:nrsr}
	H_\nr(P)\coloneqq H(\disp_\nr,\omega,v,P)\quad \mbox{and}\quad H_\sr(P) \coloneqq H(\disp_\sr,\omega,v,P).
\end{align}
We remark that we do not make additional assumptions on $\omega$ or $v$ at this point. It is noteworthy that in the literature 
the choice $\omega(k)=|k|$ and $v=\lambda\omega^{-1/2}$ with some coupling constant $\lambda\in\IR$ is often considered, due to its relevance in the physics literature, see e.g. \cite[\textsection19.2]{Spohn.2004}.

The value
\begin{align}\label{def:mw}
	\mw \coloneqq \inf_{k\in\IR^d} \omega(k)
\end{align}
is often called the {\em boson mass} in the literature. 
In the massive case $\mw>0$, the Nelson model is more infrared-regular than in the massless case $\mw=0$. To use this property, we approximate our (possibly massless) dispersion relation by a sequence of dispersion relations $\omega_n$ with positive boson mass $\mwn>0$.
Hence, also throughout this article, let us fix the following additional assumption:
\begin{enumerate}[label = (H\arabic*),ref=H\arabic*,resume]
	\item\label{hyp:wn} For all $n\in\IN$, $\omega_n:\RR^d\to [0,\infty)$ is continuous and subadditive 
	and the sequence $\omega_n(k)$ is strictly decreasing for any $k\in\IR^d$. Further, $\omega_n$ uniformly converges to $\omega$ and there exists $C>0$ such that
	\begin{equation}\label{eq:omegadiff}
		|\omega_n(k+p)-\omega_n(k)| \le C |\omega(k+p)-\omega(k)| \qquad \mbox{for all}\ k,p\in\IR^d,\ n\in\IN.
	\end{equation}
\end{enumerate}
\begin{rem} Suppose $\omega : \R^d \to [0,\infty)$   satisfies   Hypothesis \cref{hyp:omega} and \eqref{hyp:div} 
	then it is straightforward to see that   $\omega_n(k) = \sqrt{ \omega^2 + \mu_n^2 }$ satisfies Hypothesis \eqref{hyp:wn}
	for any strictly decreasing zero sequence $\mu_n$. In case $\mw=0$, this implies $\mwn=\mu_n$, so $\mu_n$ can be interpreted as an artificial boson mass.
\end{rem} 
\begin{rem} Suppose $\omega : \R^d \to [0,\infty)$   satisfies   Hypothesis \cref{hyp:omega} and \eqref{hyp:div} 
	and $\omega_n$ satisfies $ \eqref{hyp:wn}$. 
	Then it is straightforward to see  that $\inf_{k \in \R^d} \omega_n  > 0$ (even if $\mw = 0$).  
\end{rem} 
\begin{rem}
	We also remark that the assumption $\omega_n^{-1/2}v\in L^2(\IR^d)$ from \cref{hyp:v} is trivially satisfied, so the operators $H(\disp,\omega_n,v,P)$ are also selfadjoint and lower-semibounded.
\end{rem} 

\subsection{Main Results}\label{sec:res.res}

In this \lcnamecref{sec:res.res}, we present our main results on the two models introduced in \cref{def:nrsr}.

The first result concerns the existence of a dressing-transformed Hamiltonian. As a first step, we define the dressing transformation which we apply. To this end, for now assuming $\omega(k)\ge |k|$ for all $k\in\IR^d$, given $Q\in B_1(0)$ and $n\in\IN$, we  define $f_{Q,n} \in L^2(\IR^d)$ by
\begin{equation}\label{def:fQ}
	f_{Q,n}(k) \coloneqq \frac{v(k)}{\omega_n(k)-k\cdot Q} \qquad \mbox{for}\  k\in\IR^d. 
\end{equation}
\begin{rem}\label{rem:dresstrans}
	Formally, the dressing transformation we apply is given by $f_{Q,\infty}$ where we replace $\omega_n$ in above definition by $\omega$. In this case, $f_{Q,\infty}\in L^2(\IR^d)$ if and only if $\omega^{-1}v\in L^2(\IR^d)$. However, for the massless Nelson model, this is usually not satisfied.  This implies that the dressing transformation becomes non-unitary and we need to approximate it by well-defined unitary transformations. For $n\in\IN$, we remark that $f_{Q,n}\in L^2(\IR^d)$ follows, since the assumptions \cref{hyp:wn} and $\omega(k)\ge |k|$ imply $c_n \coloneqq \inf_{k\in B_1(0)}(\omega_n(k)-|k|)>0$, so $(c_n\wedge(1-|Q|))^{-1}v $ is a square-integrable majorant of $f_{Q,n}$.
\end{rem}
Our first main result is the existence of an infrared renormalized operator.
\begin{thm}\label{mainthm1}
	Let $\#=\nr$ or $\#=\sr$ and $|\mm|v\in L^2(\IR^d)$. Further, 
	assume   $\omega(k)\ge |k|$
	and let $P,Q \in \R^d$  with $|Q|<1$.
	Then
	the  operators $W(f_{Q,n})H_\#(P)W(f_{Q,n})^*$ converge to a selfadjoint and lower-semibounded operator $\wh H_\#(P;Q)$ in the norm-resolvent sense as $n\to\infty$, which is independent of the specific  sequence $(\omega_n)$ satisfying \eqref{hyp:wn}.
\end{thm}
\begin{rem}\label{rem:Hhatunitary}
	In the case $\omega^{-1}v\in L^2(\IR^d)$, the operators $H_\#(P)$ and $\wh H_\#(P,Q)$  are unitarily equivalent by the Weyl transformation $f_{Q,\infty}$ with $\omega_\infty=\omega$, cf. \cref{rem:dresstrans}, for any choice of $Q\in B_1(0)$. This is the so-called infrared-regular case. However, in the infrared-critical case $\omega^{-1}v\notin L^2(\IR^d)$, which is especially satisfied by the physical choice $\omega = |\cdot|$ and $v=\lambda\omega^{-1/2}$, $\lambda\ne 0$, the two operators are inequivalent. One then calls $\wh H_\#(P,Q)$ a non-Fock representation of the Nelson model.
\end{rem}
\begin{rem}
	Similar non-perturbative results for the non-relativistic and semi-relativistic Pauli--Fierz Hamiltonian can be found in \cite{HaslerSiebert.2022} and \cite{KoenenbergMatte.2014}, respectively.
	Our proof is essentially an adaption of the one from \cite{HaslerSiebert.2022}, especially to the semi-relativistic case.
	A direct application of the method in \cite{KoenenbergMatte.2014} does not seem possible here, since the proof therein exploits the spin degree of freedom in the Pauli--Fierz model.
\end{rem}
\begin{proof}
	The statement is a direct consequence of \cref{th:P2 diff}. More precisely:
	
	By the assumptions \cref{hyp:v} and $\omega\ge |\mm|$, the map
	\[ k \mapsto f_{Q}(k) \coloneqq \frac{v(k)}{\omega(k)-k\cdot Q} \]
	defines an element in $\sW_\#$, as defined in \cref{def:normsW}. Hence, the limiting operator $\wh H_\#(P) \coloneqq T_\#(P,f_Q)$ as defined in \cref{eq:TnPf defn} is selfadjoint and lower-semibounded, cf. \cref{firstwelldefinedT}.
	Further, by the dominated convergence theorem, we see $f_{Q,n}$ converges to $f_Q$ in $\sW_\#$.
	The statement now follows from \cref{th:P2 diff}.
\end{proof}
Our second main result now concerns the infrared behavior of the infrared renormalized operator.
It holds in 
a physically motivated momentum region, which we first want to describe.
If a non-relativistic particle interacts with a relativistic quantum field, as is the case in our model $H_\nr$, then one expects physical results only to be valid in a region where the speed of the non-relativistic particle is below the speed of light, i.e., for our choice of units, where the particle momentum is smaller than the particle mass $M>0$.
If, on the other hand, the particle has a relativistic dispersion relation, as is the case in our model $H_\sr$, no such restriction is expected.
Hence, let us
explicitly fix
\begin{align}\label{eq:relreg}
	\sB_\nr \coloneqq B_{M}(0) \quad \mbox{and}\quad \sB_\sr \coloneqq \IR^d.
\end{align}
A key assumption of our result is a mild differentiability condition on the mass shell
\begin{align*}
	E_\# (P) \coloneqq \inf \sigma(H_\#(P)) \quad \mbox{for}\ P\in\IR^d,\ \#\in\{\nr,\sr\}.
\end{align*}
Explicitly, let
\begin{align}
	\sD_{\#,k} \coloneqq \{P\in \IR^d: \xi \mapsto E_\#(\xi)\ \mbox{is \ $k$-times differentiable at}\ \xi=P\}.
\end{align}
We prove that up to the second derivative, this only excludes a set of total momenta having Lebesgue measure zero. Further, in the relevant regions \cref{eq:relreg}, we can bound the gradient.
\begin{lem}\label{lem:D}
	Let $\#\in\{\nr,\sr\}$.
	The set $\sD_{\#,k}$ has full Lebesgue measure for $k=1,2$, $\#\in\{\nr,\sr\}$. Further, for all $P\in\sD_{\#,1}\cap\sB_\#$, we have $|\nabla E_\#(P)|<1$.
\end{lem}
\begin{rem}
	In the case $\mw>0$, the mass shell is even analytic, cf. \cite[Theorem 1.9]{Moller.2005}.
\end{rem}
\begin{rem}\label{rem:AH12}
	By \cite[Theorem 1.1]{AbdessalamHasler.2012}, we know that $\sD_{\nr,k}\supset B_{1-\|\omega^{-1/2}v\|^{2/3}}(0)$ for all $k\in\IN$ in the case $\omega=|\cdot|$ and $M=1$. The authors assume that, by an adaption of the proof, a similar statement can be obtained in the semi-relativistic case.
\end{rem}
\begin{proof}
	The statement is contained in those of \cref{thm:massshell,lem:massshellsr}, noting that the second derivatives of the dispersion relations satisfy $\sfC_{\Theta_\#,2} = 1/2M$ for both cases $\#=\nr,\sr$ and that the first derivative in the semi-relativistic case satisfies $\sfC_{\Theta_\sr,1,1}=1$, $\sfC_{\Theta_\sr,1,2} = 0$.
\end{proof}
We can now state our main result.
\begin{thm}\label{mainthm2}
	Assume $\#=\nr$ or $\#=\sr$ and $|\mm|v\in L^2(\IR^d)$. Further, let $P\in \sD_{\#,2}\cap\sB_\#$ and assume $\omega(k)\ge|k|$. Then $\wh H_\#(P, \nabla E_\#(P))$ has a ground state. 
\end{thm}
\begin{rem}
	In the infrared-regular case $\omega^{-1}v\in L^2(\IR^d)$ it is well-known that $H_\#(P)$ has a ground state, see for example \cite{Frohlich.1973,Gerard.2000} as well as \cref{cor:irregulargs}, so the statement directly follows from the unitary equivalence of $H_\#(P)$ and $\wh H_\#(P)$ in this case, cf. \cref{rem:Hhatunitary}. However, in the infrared-critical case $\omega^{-1}v\notin L^2(\IR^d)$ the ground state ceases to exist, cf. \cite{Frohlich.1973,HaslerHerbst.2008a,Dam.2018,DamHinrichs.2021}. Hence, the infrared renormalized operator in fact exhibits a different infrared behavior than the original model.
\end{rem}
\begin{rem}
	In \cite{Pizzo.2003}, Pizzo proved the existence of dressed electron states in the case $\#=\nr$ for $P\in B_{M/20}$ and ($P$-dependent) sufficiently small $\|v\|$, following a construction proposed by Fr\"ohlich \cite{Frohlich.1974}. A construction of an infrared-renormalized operator for which these states are the ground state was first performed in \cite{BachmannDeckertPizzo.2012}.
	For the case $\#=\sr$, $P\in B_{M/2}(0)$ and sufficiently small $\|v\|$, a similar result was obtained in \cite{DeckertPizzo.2014}.
\end{rem}
\begin{rem}\label{rem:cont}
	In view of applications in scattering theory, it would be desirable to prove that the ground states can be chosen continuous in $P$ at least almost everywhere, cf. \cite[Eq.~(1.8)]{BeaudDybalskiGraf.2021}.
	A similar statement is derived from our compactness arguments additionally employing positivity arguments in \cref{cor:irregulargs} for a more regular case. However, in case of the infrared renormalized operators $H_\#(P)$ such positivity arguments do not seem available at this point.
\end{rem}
\begin{proof}
	We want to apply \cref{thm:compactness.infraredcrit}.
	To that end, let $H_{\#,n}(P) = H(\Theta_\#,\omega_n,v,P)$, let $E_{\#,n}(P)=\inf \sigma(H_{\#,n}(P))$ and let $\fpn = f_{\nabla E_{\#,n}(P),n}$. Note that $\fpn\in L^2(\IR^d)$ is well-defined by \cref{rem:dresstrans,lem:D,lem:perturb}.
	By \cref{prop:Moller,cor:gapnr,lem:massshellsr}, $H_{\#,n}(P)$ has a normalized ground state, which we denote by $\ppn$.
	Now, we observe that the assumptions of \cref{thm:compactness.infraredcrit} are satisfied due to \cref{thm:massshell,cor:gapnr}  in the case $\#=\nr$ and due to \cref{thm:massshell,lem:massshellsr} in the case $\#=\sr$. Hence, the set $\{W(\fpn)\ppn:n\in\IN\}$ is relatively compact. W.l.o.g., we can now assume the strong limit $\psi_P = \lim_{n\to\infty}W(\fpn)\ppn$ exists.
	
	It remains to prove that $\psi_P$ is a ground state of $\wh H_\#(P,\nabla E_\#(P))$.
	First, we note that \[\inf \sigma\left(\wh H_\#(P,\nabla E_\#(P))\right)=E_\#(P),\] by \cref{mainthm1}. Now,
	noting that $W(\fpn)\ppn\in \cD(\wh H_\#(P,\nabla E_\#(P)))$, by \cref{firstwelldefinedT,lem:transT},
	we can rewrite
	\begin{equation}
		\label{eq:rewritHhat}
		\begin{aligned}
			&0 \le \Braket{W(\fpn)\ppn,\left(\wh H_\#(P,\nabla E_\#(P))-E_\#(P)\right)W(\fpn)\ppn}
			\\  &=  
			\Braket{ W(\fpn)\ppn,\left(\wh H_\#(P,\nabla E_\#(P)) - W(\fpn)H_\#(P)W(\fpn)^*\right)W(\fpn) \ppn }
			\\ & \qquad + 
			\Braket{\ppn,\left(H_\#(P)-E_\#(P)\right)\ppn}.
		\end{aligned}
	\end{equation}
	The second expression on the right hand side of \cref{eq:rewritHhat} tends to zero as $n\to\infty$, by \cref{lem:Enconv}.
	Further, the first expression on the right hand side tends to zero, since the operators $W(\fpn)H_\#(P)W(\fpn)$ converge to $\wh H_\#(P,\nabla E(P))$ as $n\to\infty$ in  the norm resolvent sense, by \crefnosort{mainthm1,th:P2 diff,thm:massshell.dc}.
	This implies $\psi_P\in\cD(\wh H_\#(P,\nabla E_\#(P)))$.
	Further, using the lower-semicontinuity of closed quadratic forms, cf. \cite[\textsection VI, Theorem~1.16]{Kato.1980}, we find
	\begin{align*}
		0 &\le \Braket{\psi_P,\left(\wh H_\#(P,\nabla E_\#(P))-E_\#(P)\right)\psi_P} \\& \le \liminf_{n\to\infty}\Braket{W(\fpn)\ppn,\left(\wh H_\#(P,\nabla E_\#(P))-E_\#(P)\right)W(\fpn)\ppn} = 0, \end{align*}
	which finishes the proof.
\end{proof}
We especially emphasize the following implication of \cref{mainthm2}.
\begin{cor} Let $M=1$, $\omega(k) = |k|$ and let  $|P| < 1$. If 
	$
	\|\omega^{-1/2}v\|<(1-|P|)^{2/3} 
	$, then $E_\nr(\cdot)$ is analytic in a neighborhood of $P$ and
	the operator $\wh H_\nr(P,\nabla E_\nr(P))$ has a ground state. 
\end{cor} 
\begin{proof}
	Combine \crefnosort{mainthm2,rem:AH12} or \cite[Theorem 1.1]{AbdessalamHasler.2012}.
\end{proof}

\section{Compactness Theorems for Fock Spaces}\label{sec:compact}

In this \lcnamecref{sec:compact}, we derive a Fr\'echet-type compactness theorem for direct sum Hilbert spaces and combine it with the Kolmogorov--Riesz theorem to characterize compactness in Fock spaces over $L^2(\IR^d)$. Our approach is inspired by the one in \cite{HiroshimaMatte.2019,Hinrichs.2022b}.
The main result of this \lcnamecref{sec:compact} is \cref{comfock}, which is the simple sufficient compactness criterion we apply in \cref{sec:conv}. Up to this final statement all other compactness statements presented in this section are necessary and sufficient, so the statement of \cref{comfock} can be easily strengthened using these results, if necessary for other applications.

We start with the following generalized version of a  theorem going back to Fr\'echet \cite{Frechet.1908}.

\begin{thm} \label{frechet} Let $\mathcal{H} = \bigoplus_{n=1}^\infty  \mathfrak{h}_n$ with $\mathfrak{h}_n$ being a Hilbert space for all $n\in\IN$.
	Then $M \subset \mathcal{H}$ is relatively compact if and only if the following two conditions are satisfied.
	\begin{enumerate}[label = {\upshape (\roman*)}, ref = {\upshape \roman*}]
		\item\label{frechet.projection} For each $n \in \N$ there exists a  compact subset $K_n \subset \mathfrak{h}_n$ such that 
		for all $f = (f_\ell)_{\ell\in\IN} \in M$ we have $f_\ell \in K_\ell$.
		\item\label{frechet.tight} $\sup_{f \in M } \sum_{k=n}^\infty \| f_k \|^2\xrightarrow{n\to\infty} 0 $. 
	\end{enumerate} 
\end{thm}
\begin{proof}
	We first prove the if part using a standard diagonal sequence trick.
	To that end,
	let $(g_j)_{j \in \N} = ((g_{j,\ell})_{\ell\in\IN})_{j\in\IN} $ be a sequence in $M$. We have to show that it has a convergent subsequence. 
	It follows from \subcref{frechet.projection}  that $(g_{j,1})_{j \in \N}$ has a convergent subsequence
	$(g_{n^{(1)}_j,1})_{j \in \N}$. 
	Then it follows from \subcref{frechet.projection} that $(g_{n^{(1)}_j,2})_{j \in \N}$
	has a convergent subsequence $(g_{n^{(2)}_j,2})_{j \in \N}$, 
	where $n^{(2)}_j$ is a subsequence of $n^{(1)}_j$. 
	Proceeding in this way, we inductively obtain 
	a subsequence $(n^{(\ell)}_j)_{j \in \N}$ of $(n^{(\ell-1)}_j)_{j \in \N}$ such that $(g_{n^{(\ell)}_j,s})_j$ converges 
	for all $s=1,\ldots,\ell$. Thus, we can define the diagonal sequence 
	\[
	h_j \coloneqq g_{n^{(j)}_j} , \quad j \in \N .
	\]
	This is by construction a subsequence of $(g_j)_{j \in \N}$ such that $(h_{j,\ell})_{j \in \N}$ converges for each fixed $\ell\in\IN$. 
	Since $\HS$ is complete, it remains to show that this is a Cauchy sequence. 
	Let $\eps > 0$.  Then, by \subcref{frechet.tight}, there exists an $N_0 \in \N$ such that 
	\begin{align*}
		\sup_{f \in M} \sum_{k=N_0}^\infty  \| f_k \|^2   < \frac\eps8 .  
	\end{align*} 
	Further, since $(h_{j,s})_{j \in \N}$ especially is convergent for any $s=1,\ldots,N_0$, there exists an $N \in \N$ such that 
	\begin{align*} 
		\sum_{s=1}^{N_0} \|  h_{j,s} - h_{j',s} \|^2 < \frac\eps2  \qquad \mbox{for all }j,j' \geq N.
	\end{align*} 
	Combining the above two inequalities, it follows that for $j, j'  \geq N\vee N_0$ that 
	\begin{align*} 
		\sum_{s=1}^{\infty} \|  h_{j,s} - h_{j',s} \|^2  & < \frac\eps2 + \sum_{s=N_0}^\infty  \| h_{j,s} - h_{j',s} \|^2  
		\leq  \frac\eps2 + 2 \sum_{s=N_0}^\infty  ( \| h_{j,s} \|^2 + \| h_{j',s} \|^2 ) 
		\\& < \frac\eps 2 + 4 \cdot\frac\eps8 = \eps.
	\end{align*} 
	Since $\eps$ was arbitrary, this proves the if part.
	
	To prove the only if part,
	assume that $M$ is relatively compact.
	Then \cref{frechet.projection} easily follows by choosing $K_n$ to be the closure of $\{g_n\in\hs_n : \exists f=(f_\ell)_{\ell\in\IN}\in M: f_n=g_n\}$, since convergence of a sequence in $\HS$ immediately implies convergence of the corresponding sequence in the subspace $\hs_n$ for any $n\in\IN$.
	Now assume that \cref{frechet.tight} is not satisfied. Then there exists an $\eps>0$ and a sequence $(f_k)_{k\in\IN}\subset M$ such that
	\[\sum_{\ell=k}^{\infty}\|f_{k,\ell}\|^2\ge\eps \qquad \mbox{for all}\ k\in\IN.\]
	Since $M$ is relatively compact, we can w.l.o.g. (otherwise restrict to a subsequence) assume that $(f_n)$ converges to some $g\in\HS$. Hence, we can pick $N\in\IN$ such that $\|g-f_n\|^2<\eps/4$ for all $n\ge N$ and $\sum_{\ell=N}^\infty\|g_\ell\|^2<\eps/4$. This implies
	\[ \eps \le  \sum_{\ell=N}^\infty \|f_{N,\ell}\|^2 \le 2\sum_{\ell=N}^\infty\|f_{N,\ell}-g_\ell\|^2  + 2\sum_{\ell=N}^\infty \|g_\ell\|^2 < 4\cdot \frac \eps 4 = \eps,  \]
	which is a contradiction.
	Hence, \subcref{frechet.tight} must be satisfied for all relatively compact sets and the proof is complete.
\end{proof}
Let us now state a well-known generalization for square-integrable functions, going back to Kolmogorov \cite{Kolmogorov.1931}, Riesz \cite{Riesz.1933} and Sudakov \cite{Sudakov.1957}. 
\begin{thm} \label{riesz}
	Let $m\in\IN$.
	A subset $M \subset L^2(\R^m)$ is relatively compact if and only if
	\begin{enumerate}[label = {\upshape (\roman*)}, ref = {\upshape \roman*}]
		\item\label{kolmogorovriesz.cont}  Let $h \in \R^{m}$.  We have $\lim_{h \to 0 } \sup_ {f \in M } \| f( \cdot + h ) - f(\cdot ) \| = 0 $.
		\item\label{kolmogorovriesz.tight}  $\lim_{R \to \infty } \sup_ {f \in M } \| \chr{B_R(0)^\sfc} f \| = 0 $.
	\end{enumerate} 
\end{thm} 
\begin{proof} For a recent historical overview with elegant proofs see \cite{HancheOlsenHolden.2010,HancheOlsenHoldenMalinnikova.2019}.
\end{proof} 
For our application it will be useful to verify assumption \cref{kolmogorovriesz.cont} of \cref{riesz} not on the whole space, but rather on the complement of an arbitrarily small set. This is useful, since we can choose the small set to contain singularities.
A similar strategy was used in \cite{HiroshimaMatte.2019,Hinrichs.2022b}, but the characteristic function of the complement of the small sets was approximated by smooth functions. This makes it harder to verify the compactness conditions.

We now prove a version of the above \lcnamecref{riesz} in which we introduce the following smallness condition, which will be used repeatedly in the remainder of this \lcnamecref{sec:compact}.
We will work with a sequence $(A_n)_{n\in\IN}$ of open sets $A_n\subset \IR^m$, for some $m\in\IN$, satisfying
\begin{equation}\label{hyp:small}
	A_n\supset A_{n+1}, \quad \operatorname{dist}(A_{n+1},A_n^\sfc)>0 \quad\mbox{for all}\ n\in\IN \quad\mbox{and}\quad  \operatorname{vol}\Big(\bigcap_{n\in\IN} A_n\Big) = 0.
\end{equation}
\begin{thm} \label{cor1} 
	Let $m\in\IN$ and assume $D_n\subset \IR^d$, $n\in\IN$, satisfies \cref{hyp:small} with $A_n=D_n$.
	Then a subset $M \subset L^2(\R^m)$ is relatively compact if and only if
	\begin{enumerate}[label = {\upshape (\roman*)}, ref = {\upshape \roman*}]
		\item\label{cor1.cont} Let $h \in \R^{m}$. For all $n \in \N$,  we have $\lim_{h \to 0 } \sup_ {f \in M } \| \chr{D_n^\sfc} ( f( \cdot + h ) - f(\cdot ) ) \| = 0 $.
		\item\label{cor1.tight}  We have $\lim_{R \to \infty  } \sup_ {f \in M } \| \chr{B_R(0)^\sfc} f \| = 0 $ and $\lim_{n  \to \infty } \sup_ {f \in M } \| \chr{D_n} f \| = 0 $.
	\end{enumerate} 
\end{thm}
\begin{proof}
	For notational compactness, we set $D_{-1} = D_0 = \IR^d$ and define
	$\Omega_n = D_n\setminus D_{n+1}$ for $n\in\IN_0$.
	Then, by the assumptions, there exists a nullset $N\subset \IR^m$ such that $\IR^m = N \cup \left(\bigcup_{n\in\IN_0}\Omega_n\right)$, where the union is disjoint, and hence
	\begin{equation}\label{eq:L2decomp}
		L^2(\IR^m) = \bigoplus_{n =0}^\infty  L^2(\Omega_n).
	\end{equation}
	The only if part now easily follows, since \subcref{cor1.cont} directly follows from \cref{riesz} \cref{kolmogorovriesz.cont} and \subcref{cor1.tight} follows from \cref{riesz} \cref{kolmogorovriesz.tight} as well as \cref{frechet} \cref{frechet.tight} combined with above decomposition.
	
	We now prove the if part.
	First, we prove that $\|f\|$ is uniformly bounded in $f\in M$ using the method from \cite[p.~3]{HancheOlsenHoldenMalinnikova.2019}.
	To this end, we choose $R>0$ and $n\in \IN$ such that $\|\chr{B_R(0)^\sfc} f\|\le 1$ and $ \|\chr{D_n}f\|\le 1$. Further, fix $\vr\in(0,\operatorname{dist}(D_{n+1},D_n^\sfc)) $ such that $\|\chr{D_{n+1}^\sfc}(f(\cdot + h)-f)\|\le 1$ for all $h\in\IR^d$ with $|h|<\vr$. Then, for all $z \in \IR^d$ and $h\in B_\vr(0)$, we find
	\begin{align*}
		\|\chr{B_R(z)}f \| &\le \|\chr{B_R(z-h)}f\| + \|\chr{B_R(z)\setminus D_{n+1}}(f(\cdot +h)-f)\| + 2 \|\chr{D_n} f\| \\&\le \|\chr{B_R(z-h)}f\| + 3.
	\end{align*}
	Repeatedly applying this inequality, for any $N\in\IN$, this implies
	\[ \|\chr{B_R(0)}f \| \le 3N + \|\chr{B_R(-Nh)}f \|. \]
	Choosing $N$ large enough such that $B_R(-Nh)\cap B_R(0) = \emptyset$, this results in
	\[ \|f\| \le 2\|\chr{B_R(0)^\sfc}f\| + 3N \le 3N+2, \]
	which is uniform in $f\in M$.
	
	Next, we prove that $\chr{\Omega_n} M$ is relatively compact for any fixed $n\in\IN_0$. 
	Let $\chi_n:\RR^m\to[0,1]$ be a uniformly continuous function such that $\chi_n=1$ on $\Omega_n$ and $\chi_n= 0$ on $D_{n+2}\cup D_{n-1}^\sfc$.\footnote{Such a function is, e.g., given by the choice $\chi_n(x) = \operatorname{dist}(x,D_{n+2}\cup D_{n-1}^\sfc) / \operatorname{dist}(\Omega_n,D_{n+2}\cup D_{n-1}^\sfc)$,  which is Lipschitz continuous.} For $f\in M$, we use the estimate
	\begin{equation*}\label{eq:cutoffbound}
		\begin{aligned}
			\|\chi_n(\cdot+ h)f(\cdot +h ) &- \chi_n f\|  = \|(\chi_n(\cdot+h)-\chi_n)f + \chi_n (f(\cdot+ h)-f)\|
			\\&\le \|\chi_n(\cdot+h)-\chi_n\|_\infty \| f\| + \|\chr{D_{n+2}^\sfc}(f(\cdot +h ) - f)\|.
		\end{aligned}
	\end{equation*}
	Taking the supremum over all $f\in M$ and subsequently the limit $h\to 0$, the first expression on the right hand side tends to zero, by the uniform boundedness of $f\in M$ and the uniform continuity of $\chi_n$,  whereas
	the second expression on the right hand side tends to zero, by Assumption \cref{cor1.cont}.
	Hence, applying \cref{riesz}, 
	we find $\chi_n M$ is relatively compact.
	Since $\chr{\Omega_n}$ continuously maps $L^2(\IR^m)$ to $L^2(\Omega_n)$ and
	${\chr{\Omega_n}M} = {\chr{\Omega_n}\chi_n M}  \subset\chr{\Omega_n}\overline{\chi_n M}$, the set $\chr{\Omega_n}M$ is relatively
	compact as a subset of a compact set.
	
	In view of \cref{eq:L2decomp} and the second part of Assumption \cref{cor1.tight}, \cref{frechet} now implies that $M$ is relatively compact. This finishes the proof.
\end{proof} 
Applying the above statements, we obtain the following \lcnamecref{comfock} for Fock spaces.
\begin{thm}\label{comfock}  Let $M \subset \FS$ and assume $D_n\subset \IR^d$, $n\in\IN$, satisfies $D_n=\emptyset$ for all $n\in\IN$ or \cref{hyp:small} with $A_n=D_n$. Assume 
	\begin{enumerate}[label = {\upshape (\alph*)}, ref = {\upshape \alph*}]
		\item\label{comfock.L2} There exists $g \in L^2(\R^d)$ such that $\| a_k \psi \| \leq |g(k)|$ for almost all $k\in\IR^d$ and all $\psi\in M$.
		\item\label{comfock.cont} Let $p\in\R^d$. For all $n \in \N$, we have \[\lim_{p \to 0 } \sup_{\psi \in M} \int_{D_n^\sfc}  \| a_{k+p} \psi - a_k \psi \|^2 \d k = 0. \]
	\end{enumerate} 
	Then $M$ is relatively compact. 
\end{thm} 
The proof is based on the following \lcnamecref{comboRieszKolmogorov}. 
\begin{lem} \label{comboRieszKolmogorov} Let $\HS = \bigoplus_{n=1}^\infty L^2(\R^{m_n})$ with $m_n \in \N_0$ and assume that $D^{(n)}_m\subset \IR^{m_n}$, $m,n\in\IN$, satisfies, for fixed $n\in\IN$, either $D^{(n)}_m = \emptyset$ for all $m\in\IN$ or \cref{hyp:small} with $A_m = D^{(n)}_m$.
	A subset $M \subset \mathcal{H}$ is relatively compact if and only if
	\begin{enumerate}[label = {\upshape (\roman*)}, ref = {\upshape \roman*}]
		\item\label{combcor.cont}  Let $n,m \in \N$ and $h \in \R^{m_n}$. We have \[\lim_{h \to 0 } \sup_ {f \in M } \| \chr{(D^{(n)}_m)^\sfc} (f_n( \cdot + h ) - f_n(\cdot )) \| = 0 .\]
		\item\label{combcor.tight} For all $n \in \N$,  we have $$ \lim_{R \to 0 } \sup_ {f \in M } \| \chr{B_R^\sfc} f_n \| = 0  \quad 
		\text{ and } \quad \lim_{m  \to \infty } \sup_ {f \in M } \| \chr{D^{(n)}_m} f_n \| = 0 . $$
		\item\label{combcor.tail} $\displaystyle \sup_{f \in M } \sum_{k=n}^\infty \| f_k \|^2\to 0 $ as $n \to \infty$. 
	\end{enumerate} 
\end{lem}  
\begin{proof} The statement directly follows from combining \cref{frechet,,riesz,,cor1}. 
\end{proof} 
We can now present the
\begin{proof}[\textbf{Proof of \cref{comfock}}] 
	We verify that the conditions \cref{comboRieszKolmogorov} \cref{combcor.cont,,combcor.tail,,combcor.tight} follow from \cref{comfock.L2,comfock.cont},
	where we choose $m_n=d\cdot n$ and set
	\[ D^{(n)}_m  = \big((D_m^\sfc)^n\big)^\sfc = \bigcup_{j=1}^n \IR^{d(j-1)}\times D_m \times \IR^{d(n-j)}.  \]
	Since this implies $\operatorname{dist}( (D^{(n)}_{m+1})^\sfc,D^{(n)}_{m+2}) \ge \operatorname{dist}(D_{m+1}^\sfc,D_{m+2})$ and \[\bigcap_{m\in\IN}D^{(n)}_m = \bigcup_{j=1}^n \IR^{d(j-1)}\times \bigcap_{m\in\IN}D_m \times \IR^{d(n-j)} ,\] the assumptions of \cref{hyp:small} with $A_m = D^{(n)}_m$ follow.
	
	We first show that \cref{comfock.L2} implies \cref{combcor.tight,combcor.tail}.
	First, for any $n\in\IN$ and $\psi\in M$, we observe that  
	\begin{align*}
		\| \chr{B_R(0)^\sfc} &\psi^{(n)} \|^2   =  \int _{\R^{n\cdot  d} } \chr{B_R(0)^\sfc}(k)  |\psi^{(n)}(k)|^2  \d k
		\\	&
		\leq   \int _{\R^{n\cdot d} } \sum_{j=1}^n \chr{|k_j| \geq R/\sqrt{n}} (k)  |\psi^{(n)}(k)|^2  \d k \\
		&  =   \int _{\R^{n\cdot d} } n  \chr{|k_1| \geq R/\sqrt{n}} (k)  | \psi^{(n)}(k_1,\ldots,k_n)|^2  \d k_1 \cdots \d k_n
		\\ &
		=   \int _{\R^{d} }   \chr{|k| \geq R/\sqrt{n}} (k)  \| (a_{k} \psi)^{(n-1)} \|^2   \d k  \\
		&  \leq    \int _{\R^{d} }   \chr{|k| \geq R/\sqrt{n}} (k)  \| a_{k} \psi \|^2  \d k 
		\leq    \int _{\R^{d} }   \chr{|k| \geq R/\sqrt{n}}(k)  |g(k)|^2   \d k.
	\end{align*} 
	Note that final bound is independent of $\psi\in M$. Now, by dominated convergence, the right hand side converges to zero as $R\to\infty$, which proves the first assumption of \cref{combcor.tight}. By a similar argument, we find
	\[ \|\chr{D_m^{(n)}}\psi^{(n)}\|^2 \le \int_{D_m}|g(k)|^2\d k \qquad \mbox{for all}\ m,n\in\IN,\ \psi\in M, \]
	so the second assumption of \cref{combcor.tight} again follows by dominated convergence.
	To see \cref{combcor.tail}, we observe that
	\[\sup_{\psi \in M} \langle \psi , \dG(\Id) \psi \rangle = \sup_{\psi \in M}  \int_{\R^d} \| a_k \psi \|^2 dk \leq \int_{\R^d} |g(k)|^2 dk  = \| g \|^2 < \infty. \] 
	Hence, it follows that
	\begin{align*}
		\sup_{\psi\in M}\sum_{k=N } \| \psi^{(n)} \|^2 &\leq N^{-1} \sup_{\psi\in M}\sum_{k=N } n  \| \psi^{(n)} \|^2  \leq 
		N^{-1} \sup_{\psi\in M}\langle \psi , \dG(\Id) \psi \rangle \\& \leq N^{-1} \| g \|^2 \xrightarrow{N\to\infty}0, 
	\end{align*}
	which proves \cref{combcor.tail}.
	
	It remains to prove that \cref{combcor.cont} follows from \cref{comfock.cont}. We fix $m,n\in\IN$.
	First, we observe that for $h=(h_1,\ldots,h_n) \in \R^{n\cdot d}$ having absolut value sufficiently small (more explicitly, we require $|h|_\infty \le \operatorname{dist}(D_{m-1}\setminus D_{m},D_{m+1})/n$), using a telescope sum, the translation invariance and rotation symmetry of the Lebesgue measure and the permutation symmetry of both $\psi^{(n)}$ and $(D^{(n)}_{m+1})^\sfc$, we find
	\begin{align*}
		\| \chr{(D^{(n)}_m)^\sfc} &  \left( \psi^{(n)}( \cdot + h ) - \psi^{(n)}\right) \|  \\& = \left\| \chr{(D^{(n)}_m)^\sfc} \sum_{j=1}^n ( \psi^{(n)}( \cdot + \sum_{s=1}^j e_s  h_s ) - \psi^{(n)}(\cdot +  \sum_{s=1}^{j-1} e_s h_s ) )  \right\|  \\
		& \leq  \sum_{j=1}^n  \left\|\chr{(D^{(n)}_m)^\sfc} ( \psi^{(n)}( \cdot + \sum_{s=1}^j e_s h_s ) - \psi^{(n)}(\cdot +  \sum_{s=1}^{j-1} e_s h_s ) )  \right\| \\
		& \le   \sum_{j=1}^n \left\|\chr{(D^{(n)}_{m+1})^\sfc} ( \psi^{(n)}( \cdot + e_j h_j ) - \psi^{(n)}  )  \right\|
		\\& = \sum_{j=1}^n \left\|\chr{(D^{(n)}_{m+1})^\sfc} ( \psi^{(n)}( \cdot + e_1 h_j ) - \psi^{(n)}  )  \right\|.
	\end{align*} 
	Now to estimate the right hand side, we use that for $p \in \R^d$ 
	\begin{align*}
		&\|\chr{(D^{(n)}_{m+1})^\sfc} \psi^{(n)}( \cdot + e_1 p ) - \psi^{(n)} \| \\& = \int_{ (D^{(n)}_{m+1})^\sfc } | \psi^{(n)}( k_1 + p,k_2,\ldots,k_n) - \psi^{(n)}(k_1,\ldots,k_n) |^2 \d k_1 \cdots \d k_n  \\
		& \le \int_{ D_{m+1}^\sfc } n^{-1} | (a_{k+p} \psi)^{(n-1)}(k_2,\ldots,k_n) - (a_k \psi)^{(n-1)}(k_2,\ldots,k_n) |^2 \d k \d k_2  \cdots \d k_n  \\
		& \leq n^{-1} \int_{\R^d}  \| a_{k+p} \psi - a_k \psi \|^2\d k ,
	\end{align*} 
	which tends to zero uniformly in $\psi \in M$ as $p \to 0$. Combining the above estimates, this proves \cref{combcor.cont} and hence the statement follows from \cref{comboRieszKolmogorov}.
\end{proof} 
%
%
\section{The Transformed Operators}\label{sec:transformed}

In this \lcnamecref{sec:transformed}, we prove the existence of dressing transformed operators, cf. \cref{mainthm1}. The main result is the convergence theorem \cref{th:P2 diff}.

To study dressing transformed versions of the operators $H_\#(P)$, as defined in \cref{def:nrsr},
we will first define a transformed version $T_\#(P,f)$ for appropriate dressing functions $f:\IR^d\to\IC$.
We then prove that these are in fact the dressing transformed operators $W(f)H_\#(P)W(f)^*$ if $f\in L^2(\IR^d)$.
Finally, we introduce sufficient conditions for the norm resolvent convergence of $T_\#(P,f_n)$ along a sequence $(f_n)$. 

For measurable $f,g : \R^d \to \C$ such that $\overline{f} g \in L^1(\R^d)$, we define the sesquilinear form 
$$\fs(f,g) = \int \overline{f}(k) g(k) dk .
$$   
Let  $P \in \R^d$.  For any measurable $f \: \IR^d \longrightarrow \IC$ with 
\begin{align} \label{kinhprop}  \mm f , |\mm|^{1/2} f  \in L^2(\R^d) , 
\end{align} 
the transformed kinetic vector-valued operator given by
\[
\vv_P(f)  \coloneqq   P-\Pf + \ph(  \mm f )-  \mathfrak{s}(f, \mm f ) 
\]
is a densely defined $\dG(|\mm|)$-bounded operator.  
This allows us to extend  $\vv_P(f)^2$ to a selfadjoint positive operator via Friedrich's extension theorem, which will be denoted 
again by  the  symbol $\vv_P(f)^2$. For $\#\in\{\nr,\sr\}$, we  can hence define $\Theta_\#(\vv_P(f))$ via the selfadjoint functional calculus. 
For any measurable $f  \: \IR^d \to \IC$ such that \eqref{kinhprop}  and 
\begin{align} \label{kinhprop2}  \omega  f , \omega^{1/2} f  \in L^2(\R^d) 
\end{align} 
hold, we thus define 
\begin{align} 
	T_\#(P,f) &\coloneqq  \Theta_\# \left( \vv_P(f) \right) +  \dG(\omega) - \ph(  \omega f ) +  \fs( f , \omega f )  + \ph(v) - 2\Re \fs(  f , v  ) , \label{eq:TnPf defn}
\end{align}
where  the finiteness of the last expression follows since $\bar f v  \in L^1(\IR^d)$, by \cref{hyp:v,kinhprop2}.

Combining the assumptions \cref{kinhprop,kinhprop2}, we define the space
\begin{equation}\label{def:sW}
	\sW=\big\{f:\IR^d\to \IC\ \mbox{measurable} \big| |\mm| f,|\mm|^{1/2}f,\omega f,\omega^{1/2}f \in L^2(\IR^d) \big\},
\end{equation}
and (for later use) equip it with the norm
\begin{equation}
	\|f\|_\sW = \|\omega f\| + \|\omega^{1/2}f\| + \||\mm| f\| + \||\mm|^{1/2}f\| + \fs(f,v).
\end{equation}

\begin{lem} \label{firstwelldefinedT}  Let $\#\in\{\nr,\sr\}$ and let $f \in\sW$.
	The symmetric operator  $T_\#(P,f)$ is defined on the domain $\cD(\disp_\#(\Pf)) \cap \cD(\Hf)$ and lower-semibounded for all $P\in\IR^d$. 
	In particular, there exists a selfadjoint extension.
\end{lem} 
\begin{rem}
	We note that $\cD(\disp_\sr(\Pf))\supset \cD(\Hf)$ if $\omega \le C|\mm|$ for some $C>0$.
\end{rem}
\begin{proof} 
	From elementary estimates we see  $\cD(|\vv_P(f)|^2) \supset \cD(\Pf^2) \cap \cD(\Hf)$, given that both $\omega^{-1/2} |\mm| f \in L^2(\IR^d)$ and $(1+\omega^{-1/2})|\mm|^2 h \in L^2(\IR^d)$.
	Furthermore, under the same assumption, $\cD( \sqrt{\vv_P(f)^2 + M^2} ) = \cD(\abs{\vv_P(h)}) \supset \cD(\Hf)$.
	The existence of a selfadjoint extension follows from the Friedrichs extension theorem.
\end{proof} 
We now prove that the operators $T_\#(P,f)$ are in fact the dressing transformed version of $H_\#(P)$ for $f\in L^2(\IR^d)$.
\begin{lem}\label{lem:transT}
	Let $\#\in\{\nr,\sr\}$ and let $f \in L^2(\R^d)\cap\sW$. Then
	\begin{align}
		\label{eq:transformation formula}
		W(f) H_\#(P) W(f)^* &= 	T_\#(P,f)
	\end{align}
	is selfadjoint with domain  $\cD(\disp_\#(\Pf)) \cap \cD(\dG(\omega)) =  \cD(T_\#(P,f))$. 
\end{lem}
\begin{proof}
	From  \cref{lem:transform a,lem:transform dG}, we infer that on $\cD(\Hf)$ 
	\[
	W(f) (\Hf + \ph(v)) W(f)^* =  \Hf - \ph( \omega f ) +  \sc{ f , \omega f }  + \ph(v) - 2\Re \sc{  f , v }.
	\]
	Furthermore, by the same lemmas we have
	\[
	W(f) (P-\Pf)^2 W(f)^* = \vv_P(f)^2
	\]
	on $\cD(\Pf^2) \cap \cD(\Hf)$ and therefore also on the full domain by the selfadjointness of the left-hand side. Therefore, $W(f) \disp_\#(P-\Pf) W(f)^* = \disp_\#(\vv_P(f))$.
	This shows \cref{eq:transformation formula} on $\cD(\disp_\#(\Pf)) \cap \cD(\Hf)$ and thus the claim follows.
\end{proof}
We now prove the convergence theorem. In the statement, we denote the Friedrichs extension of $T_\#(P,f)$, which exists by \cref{firstwelldefinedT}, by the same symbol.
\begin{prop}
	\label{th:P2 diff}
	Let $\#\in\{\nr,\sr\}$ and assume $\omega(k)\ge C |k|$ for some $C>0$. We write
	\begin{equation}\label{def:normsW}
		\|f\|_\nr = \|f\|_\sW \qquad \mbox{and}\qquad \|f\|_\sr = \|f\|_\sW + \|(|\mm|^{3/2}\vee|\mm|^2)f\|
	\end{equation}
	for all $f\in\sW$ such that the respective right hand side is finite. The set of all such $f\in\sW$ is denoted as $\sW_\#$.
	Then there exists a continuous function $\sC_\#:(0,\infty)\to(0,\infty) $ such that for all $f,g\in\sW_\#\setminus\{0\}$
	\[ \left\| (T_\#(P,f)+\i)^{-1} - (T_\#(P,g)+\i)^{-1} \right\| \le \sC_\#(\|f\|_\#)\sC_\#(\|g\|_\#) \|f-g\|_\sW.  \] 
\end{prop}
\begin{proof}[Proof for $\#=\nr$]
	Throughout this proof, $C_*(f)$ for $f\in\sW$, $*\in\{0,1,2\}$ denotes an $f$-dependent constant which continuously depends on $\|\omega^a f\|$, $a\in\{\tfrac 12,1\}$ and $\fs(f,v)$. Its value can be calculated explicitly and does not change between different estimates.
	
	Using the resolvent identity, we obtain
	\begin{align}
		\label{eq:second resolvent}
		\begin{aligned}
			&	(T_\nr(P,f)+\i)^{-1}   -
			(T_\nr(P,g)+\i)^{-1}   \\&= (T_\nr(P,f)+\i)^{-1} (T_\nr(P,g) - T_\nr(P,f) )  ( T_\nr(P,g)+\i)^{-1} .
		\end{aligned}
	\end{align}
	Applying the operator identity $A^2-B^2 = B(A-B) + (A-B) A$, on the joint domain of $T_\nr(P,f)$ and $T_\nr(P,g)$,
	we can write
	\begin{align} \label{diffnonrelkin} 
		T_\nr(P,g) - T_\nr(P,f) &= \vv_P(f) \cdot D^{(1)}(g,f)  +  D^{(1)}(g,f) \cdot \vv_P(g)  +  D^{(2)}(g,f),
	\end{align}
	with
	\begin{align*}
		D^{(1)}(g,f) &= \vv_P(g) - \vv_P(f) =   \ph( \mm g )  - \ph( \mm f )     -  ( \fs({ g , \mm g }) -  \fs({ f , \mm f })), \\
		D^{(2)}(g,f) &= - (\ph(  \omega g )-  \phi( \omega f ) )  +  (  \fs({ g , \omega g }) -  \fs({ f , \omega f }))   \\&\qquad\qquad\qquad - 2\Re (  \fs({  g , v }) -  \fs({  f , v }) ).
	\end{align*} 
	We now estimate the terms on the right hand side of \cref{eq:second resolvent} after inserting \cref{diffnonrelkin}.
	First, using the standard bounds for field operators, cf. \cref{th:phi estimate}, and the definition \cref{eq:TnPf defn}, we find
	\begin{align}
		\label{eq:Hf Tn bound}
		\begin{aligned}
			\Hf + 1 &\leq 2 \Hf - \ph(\omega h) + \ph(v) +  1 - \frac{\|\omega h\|}{2\|\omega^{1/2}h\|} - \frac{\|v\|}{2\|\omega^{-1/2}v\|} \\& \le  2 T_\nr(P,h)  + C_0(h),
		\end{aligned}
	\end{align}
	for all $h\in\sW\setminus\{0\}$.
	Using \eqref{eq:Hf Tn bound}, the spectral theorem and operator positivity,  we find 
	\begin{align} \label{estHfintermsofT} 
		\| (\Hf + 1)^{1/2} &( T_\nr(P,h)    + \i)^{-1} \|\nonumber\\ &  =  \| ( T_\nr(P,h)  -  \i)^{-1}  (\Hf + 1)  ( T_\nr(P,h)    + \i)^{-1}\|^{1/2}  \nonumber \\
		& \leq  \| ( T_\nr(P,h)    - \i )^{-1}  ( 2T_\nr(P,h)  + C_0(h))  ( T_\nr(P,h)    + \i )^{-1}\|^{1/2} \nonumber\\
		&= \nn{ \frac{2 T_\nr(P,h) + C_0(h)}{T_\nr(P,h)^2 + 1}  }^{1/2} \nonumber\\
		&\leq C_1(h).
	\end{align}
	Likewise, we find 
	\begin{align} 
		\|  |\vv_P(h)|  &( T_\nr(P,h) + \i)^{-1} \| \nonumber\\  &=  \| ( T_\nr(P,h)  - \i)^{-1}  |\vv_P(h)|^2 ( T_\nr(P,h) + \i)^{-1}\|^{1/2} \nonumber \\
		& \leq  \| ( T_\nr(P,h)  - \i)^{-1}  ( T_\nr(P,h) + C_0(h)-1)  ( T_\nr(P,h)    + \i)^{-1}\|^{1/2} \nonumber \\
		&\leq C_2(h). \label{eq:estonvres0} 
	\end{align} 
	In view of the decomposition \eqref{diffnonrelkin}, we estimate the first term on the right hand side of \cref{eq:second resolvent} using \cref{estHfintermsofT,eq:estonvres0}  by
	\begin{align} \label{estD1}  
		&\nn{   \frac{1}{T_\nr(P,f)  + \i  } \vv_P(f) \cdot D^{(1)}(g,f )\frac{1}{T_\nr(P,g)    + \i } } \nonumber \\
		&\qquad  \leq \sum_{j=1}^d  \Big\|   \frac{1}{T_n(P,f)  + \i  } (\vv_P(f))_j     (D^{(1)}(g,f ))_j    (\Hf+1)^{-1/2}  \nonumber\\&\qquad\qquad\qquad\qquad \times (\Hf+ 1)^{1/2}  \frac{1}{T_\nr(P,g)    + \i } \Big\| \nonumber \\
		&\qquad \leq  C_1(g)  C_2(f) \sum_{j=1}^d \nn{   (D^{(1)}(g,f ))_j  (\Hf+1)^{-1/2} } .
	\end{align}
	The norms in the sum can now be estimated by $\||\mm|(g-f)\| + \|\omega^{-1/2}|\mm|(g-f)\| + \||\mm|^{1/2}(g-f)\|$, using the standard bounds for field operators from \cref{th:phi estimate} and their additivity. This in turn is bounded by $\|g-f\|_\sW$, due to the assumption $\omega\ge C|\mm|$.
	By taking the adjoint, we directly obtain the same estimate for $(T_\nr(P,f)+\i)^{-1}D^{(1)}(g,f)\cdot \vv_P(g) (T_\nr(P,g)+\i)^{-1}$.
	%
	Finally, for the last term $D^{(2)}(g,f)$ in \eqref{diffnonrelkin}, using \eqref{estHfintermsofT}, we get 
	\begin{align}  \label{secondres000} 
		& \nn{   \frac{1}{T_\nr(P,f)  + \i  } D^{(2)}(g,f)\frac{1}{T_\nr(P,g)    + \i } }   \nonumber \\
		&\qquad = \nn{  \frac{1}{T_\nr(P,f)  + \i  }  D^{(2)}(g,f) (\Hf + 1)^{-1/2}(\Hf + 1)^{1/2} \frac{1}{T_\nr(P,g)  + \i } } \nonumber  \\
		&\qquad \leq C_1(g)  \nn{ ( D^{(2)}(g,f) (\Hf + 1)^{-1/2} }  . 
	\end{align} 
	Again, using the standard bounds for field operators from \cref{th:phi estimate} and their additivity, the last expression can be bounded by $\|\omega(g-f)\|+\frac32\|\omega^{1/2}(g-f)\| + |\fs(g-f,v)|$.
	Inserting all these estimates into   \cref{eq:second resolvent,diffnonrelkin},
	we arrive at the desired claim.  
\end{proof} 
%
%
\begin{proof}[Proof for  $\#=\sr$]
	Let $C_*(f)$ for $f\in\sW$, $*\in\{0,1,2\}$ denote the same $f$-dependent constants as in the previous proof. Then, by repeating the arguments from \cref{eq:Hf Tn bound,estHfintermsofT,eq:estonvres0}, for all $h\in\sW\setminus\{0\}$, we find
	\begin{align}
		& \Hf +1 \le 2 T_\sr(P,h) + C_0(h) \label{eq:newHfTbound} \\
		& \|(\Hf+1)^{1/2}(T_\sr(P,h)+\i)^{-1}\| \le C_1(h) \label{eq:newHfTestimate} \\
		& \||\vv_P(h)|^{1/2}(T_\sr(P,h)+\i)^{-1}\| \le C_2(h) \label{eq:newvvTestimate}.
	\end{align}
	%
	%
	Again using the second resolvent identity, we obtain 
	\begin{align}
		&\frac{1}{T_\sr(P,f)  + \i  }   -
		\frac{1}{T_\sr(P,g)    + \i } 
		\nonumber\\& =  \frac{1}{T_\sr(P,f)  + \i  } (T_\sr(P,g) - T_\sr(P,f) )   \frac{1}{T_\sr(P,g)    + \i } \nonumber \\
		&= \frac{1}{T_\sr(P,f)  + \i  } ( \disp_\sr \left( \vv_P(g)   \right)  
		- \disp_\sr  \left( \vv_P(f)  \right)  + D^{(2)}(g,f) )   \frac{1}{T_\sr(P,g)    + \i } \label{eq:secondres2} 
	\end{align} 
	where we used the same notation as  in \eqref{diffnonrelkin}.
	By means of \cref{eq:newHfTestimate},
	the term involving $D^{(2)}(g,f)$ can be estimated as in \cref{secondres000}.
	To estimate the remaining expression,
	we use  the following operator  identities.
	First observe   that 
	\[
	h_n(x) =  \frac{1}{\pi} \int_0^n t^{-1/2} x^{1/2}  (t+ x)^{-1}  \d t  , \quad x > 0 
	\]
	satisfies 
	$h_n(x) \to 1$ und $0 \leq h_n(x) \leq 1$ for all $x > 0$. 
	It follows from the spectral theorem that for any selfadjoint operator $F$ with $F \geq c $, for some $c \in (0,\infty)$, 
	and  $ \psi \in D(F^{1/2})$  we have  
	\[
	F^{1/2} \psi =      \lim_{n \to \infty} h_n(F) F^{1/2} \psi =  \frac{1}{\pi} \int_0^\infty t^{-1/2} F (t+ F)^{-1}\psi  \d t    . 
	\]
	Thus  for selfadjoint operators $E$ and $F$ which are bounded below by some positive constant 
	it follows, as a strong integral on $D(E^{1/2} )  \cap D( F^{1/2})$ by simple operator identities,
	\begin{align}\label{diffroots} 
		F^{1/2} - E^{1/2}    &  = \frac{1}{\pi}  \int_0^\infty t^{-1/2} \left( F  (t + F)^{-1}    -  E  (t + E)^{-1}  \right) \d t  \nonumber  \\
		&  = \frac{1}{\pi}  \int_0^\infty t^{-1/2} \left( -  t  (t + F)^{-1}  +  t  (t + E)^{-1}  \right)   \d t  \nonumber  \\
		& =  \frac{1}{\pi} \int_0^\infty t^{1/2} (t + F)^{-1}(F-E)  (t + E)^{-1} \d t  .
	\end{align} 
	Now using this identity for $F = \disp_\sr(\vv_P(h))^2$ and $E = \disp_\sr(\vv_P(f))^2$,
	we find 
	\begin{align} 
		& 	\disp_\sr \left( \vv_P(g)   \right)  
		- \disp_\sr  \left( \vv_P(f)  \right) \nonumber \\
		&= \frac{1}{\pi} \int_0^\infty t^{1/2} (t + \vv_P(g)^2 + M^2)^{-1}(  \vv_P(g)^2 -  \vv_P(f)^2 )  (t +  \vv_P(g)^2 + M^2 )^{-1} \d t \nonumber  \\
		&= \frac{1}{\pi} \int_0^\infty t^{1/2} (t + \vv_P(g)^2 + M^2)^{-1} \label{eq:diffrootexp2}\\&\qquad\qquad\times(  \vv_P(f) \cdot D^{(1)}(g,f)  +  D^{(1)}(g,f) \cdot \vv_P(g)  )  (t +  \vv_P(g)^2 + M^2 )^{-1} \d t. \nonumber
	\end{align} 
	Thus, to estimate the $\disp_\sr \left( \vv_P(g)   \right)  
	- \disp_\sr  \left( \vv_P(f)  \right) $ term in \eqref{eq:secondres2}, we have to bound
	\begin{align*}
		( T_\sr(P,f) + \i)^{-1} (\vv_P(f)^2 + & M^2+t)^{-1} \vv_P(f) \cdot   D^{(1)}(g,f)  \\ & \times  (\vv_P(g)^2 + M^2+t)^{-1} ( T_\sr(P,g) + \i)^{-1}
	\end{align*} 
	and, in particular, show that these expressions decays faster than $t^{-3/2}$ for $t \to \infty$ such that the integral obtained from inserting \eqref{eq:diffrootexp2} into \cref{eq:secondres2} converges.  
	
	Our approach is similar to the proof in the non-relativistic case. 
	We  estimate the factors  involving components of $\vv_P(h)$  by 
	\begin{align*} 
		& \| |\vv_P(h)| (\vv_P(h)^2 + M^2 + t)^{-1} (T_\sr(P,h)+\i)^{-1}  \|^2 \nonumber \\
		&   = \| (T_\sr(P,h)-\i)^{-1} (\vv_P(h)^2 + M^2 + t )^{-1} \nonumber\\&\qquad\qquad\qquad \times \vv_P(h)^2 (\vv_P(h)^2 + M^2 + t )^{-1}(T_\sr(P,h)+\i)^{-1} \| \nonumber \\
		& \le   \| (T_\sr(P,h)-\i)^{-1}|\vv_P(h)|^{1/2}   \abs{\vv_P(h)} \nonumber\\&\qquad\qquad\qquad\times (\vv_P(h)^2 + M^2 + t )^{-2} |\vv_P(h)|^{1/2} (T_\sr(P,h)+\i)^{-1} \| \nonumber  \\
		& \leq \sup_{\lambda \geq 0} \frac{\lambda }{(\lambda^2 + M^2 + t)^2} \| |\vv_P(h)|^{1/2} (T_\sr(P,h)+\i)^{-1} \|^2 \\ 
		& \leq \frac{C_2(h)}{ 4 (M^2 + t)^{3/2}},
	\end{align*}
	where we used  $\sup_{\lambda \geq 0} \frac{\lambda}{(\lambda^2  + M^2 + t)^2} =   \frac{1}{4 (M^2 + t)^{3/2} }$ and \cref{eq:newvvTestimate} in the last step.
	Note that,  by \cref{th:phi estimate},
	the operator $(D^{(1)}(g,f))_j  \dGw^{-1/2}$ is bounded by a constant times $\|g-f\|_\sW$.  It remains to show
	\begin{align} \label{boundonfieldop}
		\| \dGw^{1/2}  (\vv_P(h)^2 + M^2 +t)^{-1} ( T_\sr(P,h)    + 1)^{-1} \| \
	\end{align}
	is bounded and decays faster than $t^{-3/4}$ for $t \to \infty$ . 
	To this end, using \cref{eq:newHfTbound}, we estimate 
	\begin{equation} \label{eq:v2M2t estimate 1}
		\begin{aligned}
			&( \vv_P(h)^2 + M^2 + t  )  ( 2T_\sr(P,h) + C_0(h))   ( \vv_P(h)^2 + M^2 + t  )\\
			&  \geq ( \vv_P(h)^2 + M^2 + t  ) \dGw( \vv_P(h)^2 + M^2 + t  )  \\
			&  =\vv_P(h)^2 \dGw \vv_P(h)^2 + (M^2 + t)^2 \dGw \\ &\qquad  + (M^2 + t)( \dGw \vv_P(h)^2 + \vv_P(h)^2\dGw ).
		\end{aligned}
	\end{equation}
	The last term is a sum from $j=1$ to $d$ with summands 
	\begin{align*}
		&(\dG(\omega\wedge1)+1) \vv_P(h)_j^2 + \vv_P(h)_j^2 (\dG(\omega\wedge1)+1)  \\ &= 2 \vv_P(h)_j (\dG(\omega\wedge1)+1) \vv_P(h)_j  \\&\qquad\qquad + [\dG(\omega\wedge1),\vv_P(h)_j] \vv_P(h)_j + \vv_P(h)_j [\vv_P(h)_j,\dG(\omega\wedge1)] \\
		&\geq [\dG(\omega\wedge1),\vv_P(h)_j] \vv_P(h)_j  + \vv_P(h)_j [\vv_P(h)_j,\dG(\omega\wedge1)] \\
		&=  [[\dG(\omega\wedge1),\vv_P(h)_j], \vv_P(h)_j ]  \\
		&= [[\dG(\omega\wedge1), \ph(\i \mm_j h)], -\Pfi j + \ph(\i \mm_j h) ] \\
		& = [-\i \ph((\omega\wedge1)\mm_jh),-\Pfi{j}+\ph(\i\mm_j h)] \\
		&=  \ph(\i (\omega\wedge1) \mm_j^2 h  )  - 2 \Re \sc{ (\omega\wedge1) \mm_j h,\mm_j h  } \\
		&\geq - \epsilon \dG(\omega\wedge1) - C_\sr(h) C_\epsilon,
	\end{align*}
	for any $\epsilon > 0$ for some constants $C_\epsilon$ and $C_\sr(h)$, where the latter depends continuously on the $L^2$-norms of $|\mm|^2 h$, $\omega^{-1/2} |\mm|^2 h$, $\omega^{-1/2} |\mm| h$. In the last step we used that the field operator is infinitesimally bounded by $\dG(\omega\wedge1)$, cf. \cref{th:phi estimate}. Using this in \eqref{eq:v2M2t estimate 1} gives 
	\begin{align*}
		( \vv_P(h)^2 + M^2 + t  )  & ( 2T_\sr(P,h) + C_0(h))   ( \vv_P(h)^2 + M^2 + t  ) \\ &\geq  (M^2 + t - \epsilon)^2 \dGw - (M^2 + t) C_\sr(h) C_\epsilon ,
	\end{align*}
	and thus,
	\begin{align*}
		& ( T_\sr(P,h)    - \i)^{-1}  (\vv_P(h)^2 + M^2 +t)^{-1} \dGw \\&\qquad\qquad\qquad\qquad\qquad\times (\vv_P(h)^2 + M^2 +t)^{-1} ( T_\sr(P,h) + \i)^{-1} \\
		&\leq \frac{1}{(M^2 + t - \epsilon)^2} \frac{2T_\sr(P,h) + C_0(h)}{T_\sr(P,h)^2 +1} \\& \qquad \qquad  + \frac{C_\sr(h) C_\epsilon}{M^2 +t} ( T_\sr(P,h)    - \i)^{-1} (\vv_P(h)^2 + M^2 +t)^{-2}  ( T_\sr(P,h)    + \i)^{-1} \\
		&\leq \frac{\sup_\lambda (\frac{2 \lambda + C_0(h)}{\lambda^2 + 1})}{(M^2 + t - \epsilon)^2}  +  \frac{C_\sr(h) C_\epsilon}{(M^2 +t)^3}.
	\end{align*}
	Finally, we can choose any $\epsilon \in (0,M^2)$. 
	This implies that \eqref{boundonfieldop} indeed decays as $t^{-1}$ for $t \to \infty$ and therefore finishes the proof.
\end{proof} 

\begin{rem} {  We note that we can generalize the above proof to more general $\Theta$ by modifying the integrand 
		in 
		\eqref{diffroots},   as long as the decay as $t\to\infty$ is sufficient such that the integrals converge. } 
\end{rem} 
%
%
\section{Analysis of Ground States}\label{sec:conv}

In this \lcnamecref{sec:conv}, we study the mass shell of arbitrary Nelson-type models as defined in \cref{def:general}. To avoid notational overload, we will write
\begin{align*}
	H(P) \coloneqq H(\disp,\omega,v,P) \quad \mbox{and}\quad E(P) \coloneqq \inf \sigma (H(P)) \quad\mbox{for}\ P\in\IR^d
\end{align*}
throughout this \lcnamecref{sec:conv}.
When approximating $\omega$ by a sequence $(\omega_n)_{n\in\IN}$ of dispersion relations, as given by \cref{hyp:wn}, we also write
\begin{align*}
	H_n(P) = H(\disp,\omega_n,v,P) \quad \mbox{and}\quad  E_n(P) = \inf \sigma (H_n(P)) \quad \mbox{for}\ P\in\IR^d.
\end{align*}
In this \lcnamecref{sec:conv}, we proceed as follows: In \cref{sec:conv.unique}, we prove uniqueness and positivity of ground states independent of the boson mass $\mw$, as defined in \cref{def:mw}. In \cref{sec:conv.massive}, we then briefly discuss the case $\mw>0$, by recalling well-known existence results and calculating the first two derivatives of $E(\cdot)$ in a region, where the mass shell is evidently analytic due to the spectral gap. 
Afterwards, we study the Nelson-type models with $\mw=0$ in terms of the approximating sequence given by \cref{hyp:wn}.  In \cref{sec:conv.infrared}, we derive infrared bounds which are the main ingredient of proof for the compactness statements given in \cref{sec:conv.approx}. In the final part of this \lcnamecref{sec:conv}, we then discuss regularity properties of the mass shell essential to verify the assumptions of the compactness result \cref{thm:compactness}.
We emphasize that the results of \cref{sec:conv.approx,sec.conv:ex} especially imply \cref{mainthm2}.

\subsection{Simple Properties of the Mass Shell}
\label{sec:conv.simple}
We will employ the following well-known statements.
\begin{lem}\label{lem:Esimpleprops}\ 
	\begin{enumlem}
		\item\label{lem:massshellcont} The map $P\mapsto E(P)$ is continuous.
		\item\label{lem:massshelldiv} If $|P|\to\infty$, then $E(P)\to \infty$.
	\end{enumlem}
\end{lem}
\begin{proof}
	Statement \subcref{lem:massshellcont} is contained in that of \cite[Lemma~3.2(3)]{Dam.2018}. Statement \subcref{lem:massshelldiv} can be inferred by a simple alteration of the arguments in \cite[\textsection15.2]{Spohn.2004}. Explicitly, by the subadditivity assumption \cref{hyp:omega}, we have $\Hf\ge \omega(\Pf)$ and hence, for any $\eps>0$ we find a constant $C_\eps > 0$ such that for all $\psi\in\cD(H(P))$,
	\begin{align*} \braket{\psi,H(P)\psi} &\ge \braket{\psi,(\disp(P-\Pf) + (1-\eps)\Hf)\psi} - C_\eps\|\psi\|^2 \\&\ge  \braket{\psi,(\disp(P-\Pf) + (1-\eps)\omega(\Pf))\psi} - C_\eps\|\psi\|^2 . \end{align*}
	Assumption \cref{hyp:div} now directly yields the statement.
\end{proof}
The following properties can easily be inferred by standard arguments.
\begin{prop}Let $P\in\IR^d$.
	\begin{enumprop}
		\item\label{lem:operatormon} $H(P)\le H_n(P)\le H_{n'}(P)$ for all $n\le n'$.
		\item\label{lem:Enconv} $E(P)=\lim_{n\to\infty}E_n(P)$.
		\item\label{lem:minimizingsequence} Assume the sequence $(\ppn)_{n\in\IN}\in \cD(H(P))$
		satisfies \[\lim_{n\to\infty}\braket{\ppn,(H_n(P)-E_n(P))\ppn}=0.\] Then $\lim_{n\to\infty}\braket{\psi_{P,n},(H(P)-E(P))\psi_{P,n}} = 0$.
	\end{enumprop}
\end{prop}
\begin{rem}
	It is noteworthy, although the authors are not aware of any applications in the literature yet, that the approximating sequence in \subcref{lem:minimizingsequence}  need not contain any ground states of $H_n(P)$ but merely low energy states, which in contrary to ground states always exist. A possible application of this simple observation is that one might use arbitrary low-energy states of regularized Hamiltonians to prove existence of ground states and can hence omit the HVZ-type arguments necessary to prove \cref{prop:Moller}, also see \cite{Frohlich.1973,Moller.2005}.
\end{rem}
\begin{proof} The proof of \subcref{lem:operatormon} and \subcref{lem:Enconv} is a word by word transcription of the proof of \cite[Lemma~3.1]{HaslerHinrichsSiebert.2021a}. Statement \subcref{lem:minimizingsequence} now easily follows, similar to \cite[Proposition~3.3]{HaslerHinrichsSiebert.2021a}, by combining the estimate
	\[ 0\le \braket{\ppn,(H(P)-E(P))\ppn} \le \braket{\ppn,(H_n(P)-E(P))\ppn}\!,  \]
	which follows from \subcref{lem:operatormon}, with \subcref{lem:Enconv} and the assumption.
\end{proof}

\subsection{Uniqueness of Ground States}\label{sec:conv.unique}
Let us first recall the following well-known uniqueness result.	It is standard for real-valued $v$ which are non-zero almost everywhere, cf. \cite[Theorem 1.3]{Moller.2005}. It was extended to the case of rotation-invariant $\disp$, $\omega$ and $v$ in \cite[Lemma 4.5]{Dam.2018}.
Our statement is a simple reformulation of Dam's proof, as we illustrate below.

Following \cite[Section 6]{DamMoller.2018b}, we introduce a notion of positivity in $\FS$, which is induced by a measurable function $f:\IR^d\to \IC$. We define the convex cone
\begin{align*}
	&\sP\!\!_f = \big\{\psi=(\psi^{(n)})\in\FS\big  |   \ \psi^{(0)}\ge0, \\ & \qquad \qquad \ \bar{f(k_1)\cdots f(k_n)}\psi^{(n)}(k_1,\ldots,k_n)\ge 0\ \mbox{for a.e.}\ (k_1,\ldots,k_n)\in\IR^{d\cdot n}\big\}.
\end{align*}
\begin{prop}\label{prop:uniqueness}
	Assume that
	\begin{equation}\label{eq:nondegassump}
		{E(P-k)+\omega(k)}> E(P) \qquad \mbox{for all}\ k\notin\{\omega=0\}.
	\end{equation}
	If $E(P)$ is an eigenvalue of $H(P)$, then the corresponding eigenspace is spanned by a unique  normalized element of $\pos$.
\end{prop}
\begin{rem}
	For the non- and semi-relativistic Nelson model, \cref{eq:nondegassump} is satisfied for all $P\in \sB_\#$ provided that $\omega\ge|\cdot|$. This immediately follows from \cref{cor:gapnr,lem:massshellsr}.
\end{rem}
\begin{rem}
	Usually in the literature, the statement is stronger in the sense that the unique element of $\pos$ is even {\em strictly} positive. This, however, is not the case if we do not assume $v\ne0$ almost everywhere. In fact, our below proof shows that all $n$-particle states of the ground state vanish on the set $\{v=0\}$, see \cite[Theorem~4.5]{DamMoller.2018b} for a similar statement on the spin boson model.
\end{rem}
\begin{proof} We proceed exactly as in the proof of \cite[Lemma 4.5]{Dam.2018} and refer the reader there for more details.
	First, we note that the statement is well-known in case $v\ne 0$ almost everywhere, cf. \cite[Lemma 6.2]{DamMoller.2018b}.
	
	Now,
	let $M_v = \{v\ne 0\}$ and let $U_v$ denote the natural unitary mapping
	\[ \FS \to \FS(M_v)\oplus \bigoplus_{n=1}^\infty L^2((M_v^\sfc)^n,\FS(L^2(M_v))),  \]
	see for example \cite[Appendix A]{Dam.2018} for an explicit construction.
	Further, let $\wt H(P)$ be the operator acting on $\FS(M_v)$ defined by
	\[ \wt H(P) = \disp(P-\dG(\mm \chr{M_v})) + \dG(\omega\chr{M_v}) + \ph ( v)\] 
	and note that $\inf \sigma (\wt H(P))\ge E(P)$, cf. \cite[Lemma A.3]{Dam.2018}.
	We assume that $\psi$ is a ground state of $H(P)$ and write $U_v\psi = (\wt\psi^{(n)})_{n\in\IN_0}$. Then,
	by a direct calculation of $U_vH(P)U_v^*$, for $n\in\IN$, it follows that $\wt\psi^{(n)}$ is a ground state of the operator acting as
	\begin{align*}
		\left(\wt H(P-k_1-\ldots-k_n)+\omega(k_1)+\cdots+\omega(k_n)\right) \wt\psi^{(n)}(k_1,\ldots,k_n)
	\end{align*}
	for almost all $k_1,\ldots,k_n\in M_v^\sfc$.
	Hence, if $\psi^{(n)}$ is non-zero for any $n\in\IN$, then $E(P)$ is an eigenvalue of 
	$\wt H(P-k_1-\ldots-k_n)+\omega(k_1)+\cdots+\omega(k_n)$ for $k_1,\ldots,k_n\in M_v^\sfc$ in a set of positive $n\cdot d$-dimensional Lebesgue measure.
	However, by the subadditivity of $\omega$ and \cref{eq:nondegassump}, this can only be the case if $\omega(k_1+\cdots+k_n)=0$, which is a contradiction to \cref{hyp:omega}.
	
	Hence, $U\psi = (\wt\psi^{(0)},0,\cdots)$, where $\wt\psi^{(0)}$ is a ground state of $\wt H(P)$. This directly implies that $\wt \psi^{(0)}$ must be a complex multiple of a $\pos$-positive vector from $\FS(M_v)$. Since $U$ can also easily be observed to preserve positivity w.r.t. $\pos$, this finishes the proof.
\end{proof}

\subsection{The Massive Case}\label{sec:conv.massive}
In this \lcnamecref{sec:conv.massive}, we recall results for the massive case $\mw>0$ and analyze the infrared behavior of the ground states in this case.
\subsubsection{Existence of Ground States}
To make use of the results in \cite{Moller.2005}, we introduce the notation
\begin{equation}\label{eq:cI0}
	\cI_0(\disp,\omega,v) = \{ P\in\IR^d : E(P) < \inf \{ E(P-k) + \omega(k) : k\in\IR^d \}  \}.
\end{equation}
The result of M\o ller for the massive case $\mw>0$, cf. \cref{def:mw}, is the following. The proof uses an HVZ type theorem, which in fact shows that $\inf \{ E(P-k) + \omega(k) : k\in\IR^d \}$ is the infimum of the essential spectrum of $H(P)$, also see \cite{Frohlich.1973} for a similar statement.
\begin{prop}[{\cite[Corollary 1.4]{Moller.2005}}]\label{prop:Moller}
	Assume $\mw >0$. $E(P)$ is a discrete eigenvalue of $H(P)$ if $P\in \cI_0(\disp,\omega,v)$.
\end{prop}
\begin{rem}
	The results in \cite{Moller.2005} are even stronger. In fact, M\o ller proved that if the subadditivity in \cref{hyp:omega} is strict, then the reverse implication holds as well, cf. \cite[Theorem 1.5]{Moller.2005}.
\end{rem}
\subsubsection{Analytic Perturbation Theory for the Total Momentum}
Since $E(P)$ is a (simple) discrete eigenvalue for $P\in\cI_0(\disp,\omega,v)$, by \cref{prop:Moller,prop:uniqueness}, we can apply analytic perturbation theory in the total momentum to calculate its derivatives.
\begin{lem}\label{lem:perturb}
	Assume $\mw>0$ and $P\in \cI_0(\disp,\omega,v)$. Then $\xi \mapsto E(\xi)$ is analytic for $\xi$ in a neighborhood of $P$. Further, for an arbitrary normalized ground state $\psi_P$ of $H(P)$, the following holds:
	\begin{enumlem}
		\item\label{lem:perturb.1} $\partial_i E(P) = \braket{\psi_P,\partial_i\disp(P-\Pf)\psi_P}$ for all $i=1,\ldots,d$.
		\item\label{lem:perturb.2} For all $i=1,\ldots,d$,
		\[ \partial_i^2 E(P) \le 2\CT2 -  2\left\|(H(P)-E(P))^{-1/2} \left(\partial_i\disp(P-\Pf)-\partial_i E(P))\right) \psi_P  \right\|^2. \]
	\end{enumlem}
\end{lem}
\begin{proof}
	First, we observe that
	\begin{align*}
		\IC^d\ni \eta
		\mapsto
		\disp(P-\eta-\Pf) + \dG(\omega) + \ph(v) =: {h_P(\eta)}
	\end{align*}
	is an analytic family of type (A), by the analyticity of $\disp$, cf. \cite{Kato.1980}.
	Hence, by \cref{prop:Moller,prop:uniqueness}, there exist analytic maps $e_{P}$ and $\phi_{P}$ defined in a neighborhood of zero such that $\phi_{P}(0)=\psi_P$ and $(h_{P}(\eta)-e_{P}(\eta))\phi_{P}(\eta)=0$. We remark that this directly implies $e_P(0)=E(P)$.
	Taking the first partial derivative w.r.t. $\eta_i$ and taking the scalar product with $\psi_P$ now directly yields \subcref{lem:perturb.1}. Further, \subcref{lem:perturb.1} implies that $(\partial_i \disp (P-\Pf) - \partial_i E(P))\psi_P\in \ker(H(P)-E(P))^\perp$. Hence, we can solve for the derivative of the eigenvector and obtain \[\partial_i \phi_{P}(0) = -(H(P)-E(P))^{-1}(\partial_i \disp (P-\Pf) - \partial_i E(P))\psi_P + \alpha \psi_P \qquad \mbox{for some}\ \alpha\in\IC.\]  Again differentiating w.r.t. $\eta_i$, taking the scalar product with $\psi_P$ and inserting above identity for $\partial_i \phi_{P}(0)$, we obtain
	\begin{align*}\partial_i^2 E(P) = & \Braket{\psi_P,\partial_i^2 \disp (P-\Pf)\psi_P}\\& -2\|(H(P)-E(P))^{-1/2}\left(\partial_i\disp(P-\Pf)-\partial_i E(P))\psi_P\right)\|^2 \end{align*}
	Hence, the statement follows using \cref{hyp:Omega}.
\end{proof}

\subsubsection{Infrared Bounds}\label{sec:conv.infrared}
\begingroup\renewcommand{\Rnp}[1]{R(P,#1)}
In this \lcnamecref{sec:conv.infrared}, we derive bounds on the infrared behavior of the resolvents
\begin{equation}
	\Rnp{k} = (H(P-k) - E(P) + \omega(k))^{-1}.
\end{equation}
We remark that these are well-defined for any $P\in\cI_0(\disp,\omega,v)$, by the definition \cref{eq:cI0}.

We start out with the following simple statement, which will be useful later on.
\begin{lem}\label{lem:I0char}
	If $\mw>0$, then $\displaystyle\inf_{k\in\IR^d}\frac{E(P-k)-E(P)+\omega(k)}{\omega(k)}>0 $ if and only if $P\in\cI_0(\disp,\omega,v)$.
\end{lem}
\begin{proof}
	By \cref{lem:Esimpleprops,hyp:omega}, there exists $k_0\in\IR^d$ such that \[\inf_{k\in\IR^d}\frac{E(P-k)-E(P)+\omega(k)}{\omega(k)}= \frac{E(P-k_0)-E(P)+\omega(k_0)}{\omega(k_0)}.\] Hence, the statement follows from the definition \cref{eq:cI0}.
\end{proof}
\Cref{lem:I0char} allows us to define
\begin{align}\label{def:DP}
	\DP = \sup_{k\in\IR^d}\left(\frac{\omega(k)}{E(P-k)-E(P)+\omega(k)}\right)\in(0,\infty)
\end{align}
for $P\in\cI_0(\disp,\omega,v)$.
For the remainder of this \lcnamecref{sec:conv.infrared}, we use the more compact notation $\Delta_P=\DP$.
\begingroup\renewcommand{\DP}{\Delta_P}
The next simple statement is an a priori bound, which we frequently use in the following and which directly entails existence of ground states in the infrared-regular case, cf. \cref{cor:irregulargs}.
\begin{lem}\label{lem:apriori}
	Assume $\mw>0$ and $P\in\cI_0(\disp,\omega,v)$. Then
	\[\| \Rnp k\| \le \frac{\DP}{\omega(k)}  \qquad \mbox{for all}\ k\in\IR^d.\]
\end{lem}
\begin{proof}
	The statement directly follows from \cref{def:DP} and the spectral theorem.
\end{proof}
Recalling the definition \cref{def:fQ}, the remaining discussion in this \lcnamecref{sec:conv.infrared} is devoted to studying the (small) $k$-dependence of $v(k)\Rnp k - f_{\nabla E(P),n}(k)$ on the eigenspace corresponding to $E_n(P)$. We start out with the following simple operator bound.
\begin{lem}\label{lem:opdiff}
	Assume $\mw>0$ and let $P\in\cI_0(\disp,\omega,v)$. Then, for all $k\in\IR^d$,
	\[ \left\|\left(H(P)-E(P)\right)\Rnp k\right\| \le 1+\CT{1,2}|k| + \DP(\CT{1,1} + \CT2 |k|)|k|\omega(k)^{-1}.\]
\end{lem}
\begin{proof}
	By \cref{lem:dispdiff}, we have
	\begin{align*}
		H(P) - E(P) & = H(P-k) - E(P) + k\cdot \nabla \disp(P-k-\Pf) + D_P(k),
	\end{align*}
	on $\cD(H_n(0))$, where $D_P(k)\in\cB(\FS)$ with $\|D_P(k)\|\le \CT{2}|k|^2$.
	To estimate the gradient term on the right hand side, we first use \cref{lem:dispderbound}) and afterwards the trivial bound $\disp (P-k-\Pf) \le H(P-k)-E(P)+\omega(k)$. Multiplying with $\Rnp{k}$ from the right, taking the operator norm and using the a priori bound \cref{lem:apriori} yields the statement.
\end{proof}
\begin{lem}\label{lem:perturbation estimate}
	Assume $\mw>0$, let $P\in \cI_0(\disp,\omega,v)$ and let  $\psi_P$ be a normalized ground state of $H(P)$.
	Then, for all $k\in\IR^d$,
	\begin{align*}
		&\left\|\Rnp k\left(\partial_i \disp(P-\Pf)-\partial_i E(P)\right)\psi_P\right\|
		\\& \qquad 
		\le |\CT2-\tfrac12\partial_i^2 E(P)|^{1/2} \\&\qquad\qquad\times\left(1+\CT{1,2}|k| + \DP(\CT{1,1} + \CT2 |k|)|k|\omega(k)^{-1}\right)^{1/2}\DP^{1/2}\omega^{-1/2}(k).
	\end{align*}
\end{lem}
\begin{proof}
	The product inequality yields
	\begin{align*}
		&\left\|\Rnp k\left(\partial_i \disp(P-\Pf)-\partial_i E(P)\right)\psi_P\right\|
		\\ & \
		\le \|\Rnp{k}^{1/2}\| \|\Rnp{k}^{1/2}(H(P)-E(P))^{1/2}\|\\&\qquad\qquad\times\|(H(P)-E(P))^{-1/2}\left(\partial_i \disp(P-\Pf)-\partial_i E(P)\right)\psi_P\|.
	\end{align*}
	Hence, combining \cref{lem:opdiff,lem:perturb.2,lem:apriori} proves the statement.
\end{proof}

\begin{lem}\label{lem:resapprox}
	Assume $\mw>0$. Let $P\in\cI_0(\disp,\omega,v)$, let $\psi_P$ be a normalized ground state of $H(P)$ and
	let \begin{equation}\label{def:Cw}\Cw = \sup_{k\in\IR^d}\frac{|k|}{\omega(k)}\in(0,\infty].\end{equation}
	If $|\nabla E(P)|<\Cw^{-1}$,
	then there exists a constant $C>0$, solely depending on the choice of $\disp$, such that for all $k\in\IR^d$
	\begin{align*}
		&\Big\|\big(\Rnp k - \frac{1}{\omega(k)-k\cdot \nabla E(P)}\big) \psi_{P}\Big\| \\& \quad\le C \sum_{i=1}^d \left(1+|\CT2-\tfrac12\partial_i^2 E(P)|^{1/2}\right)\frac{(\Cw\vee \Cw^2)(1\vee \DP)}{1-\Cw|\nabla E(P)|}\left(1+\omega(k)^{-1/2}\right).
	\end{align*}
	We emphasize that $C$ is independent of $\omega$, $v$ and $P$.
\end{lem}
\begin{proof}
	\newcommand\Rznp[1]{R^{(0)}(P,#1)}
	We write
	\[\Rznp k \coloneqq (H(P)-E(P)+\omega(k))^{-1}.\]
	Using the resolvent identity and \cref{lem:dispdiff}, there exists $D_P(k)\in\cB(\FS)$ for all $k\in\IR^d$ with $\|D_P(k)\|\le \CT{2}|k|^2$ such that
	\begin{align*}
		\Rnp k = & \Rznp k + \Rnp k \left[k\cdot \nabla \Theta(P-\Pf)\right] \Rznp k \\&+  \Rnp k D_P(k) \Rznp k.
	\end{align*}
	Applying this to $\omega(k)\psi_P$ and using $\Rznp k \psi_P = \omega(k)^{-1}$ yields
	\begin{align*}
		\omega(k)\Rnp k \psi_P  =  & \psi_P + (k \cdot \nabla E(P)) \Rnp k \psi_P \\ & + \Rnp k \left(k\cdot (\nabla \disp(P-\Pf)-\nabla E(P))\right)\psi_P \\ & + \Rnp k D_P(k) \psi_P.
	\end{align*}
	Since $\omega(k)-k\cdot \nabla E(P)>0$ for all $k\in\IR^d$, by assumption, we can rearrange to obtain
	\begin{align}
		\big(\Rnp k &- \frac{1}{\omega(k)-k\cdot \nabla E(P)}\big) \psi_{P} \nonumber\\
		= &\frac{1}{\omega(k)-k\cdot \nabla E(P)} \Rnp k \left[ k \cdot \left(\nabla \disp(P-\Pf) - \nabla E(P)\right)\right] \psi_P \label{eq:term1}\\
		& + \frac{1}{\omega(k)-k\cdot \nabla E(P)} \Rnp k D_P(k) \psi_P.\label{eq:term2}
	\end{align}
	Combining \cref{lem:perturbation estimate} with the estimate
	\begin{align}\label{eq:fctbound}
		\frac{|k|}{\omega(k)-k\cdot \nabla E(P)} \le \frac{\Cw}{1-\Cw |\nabla E(P)|},
	\end{align}
	we find
	\[ \|\eqref{eq:term1}\| \le C \sum_{i=1}^d |\CT2-\tfrac12\partial_i^2 E(P)|^{1/2}\frac{(\Cw\vee\Cw^2)(\DP^{1/2}\vee\DP)}{1-\Cw|\nabla E(P)|}(1+\omega(k)^{-1/2} ).\]
	Further, by \cref{lem:apriori,eq:fctbound}, we find
	\[  \|\eqref{eq:term2}\| \le C\frac{\Cw^2\DP}{1-\Cw |\nabla E(P)|}.\]
	Combining these estimates proves the statement.
\end{proof}
\endgroup\endgroup

\subsection{Approximation by Massive Models}\label{sec:conv.approx}
Let us now move to the case $\mw=0$. 
The strategy is as follows: We approximate $\omega$ by $(\omega_n)_{n\in\IN}$ and combine the compactness theorem from \cref{sec:compact} with the bounds from \cref{sec:conv.infrared} to deduce compactness of the set of ground states of $H_n(P)$ and $W(\fpn)H_n(P)W(\fpn)^*$ in the infrared-regular and infrared-critical case, respectively.

The connection between the resolvents
\begin{equation}
	\Rnp{k} = (H_n(P-k) - E_n(P) + \omega_n(k))^{-1}.
\end{equation}
estimated in \cref{sec:conv.infrared} and the pointwise annihilators acting on ground states is given by the so-called pull-through formula.
Formally, it directly follows from the canonical commutation relations. Rigorous proofs can be found throughout the literature, see e.g. \cite{Frohlich.1973,BachFroehlichSigal.1998b,Gerard.2000,Dam.2018}.
\begin{lem}\label{lem:pullthrough}
	Fix $n\in\IN$ and $P\in\cI_0(\disp,\omega_n,v)$.  
	If $\ppn\in\chr{\{E_n(P)\}}(H_n(P))\FS$,
	then $a_k\ppn\in \FS$ for almost all $k\in\IR^d$ and in this case
	\begin{equation}\label{eq:standardpt}
		a_k\ppn = -v(k)\Rnp k\ppn.
	\end{equation}
	Further, $a_kW(\fpn)\ppn \in \FS$ for almost all $k\in\IR^d$ and in this case
	\begin{equation}
		\label{eq:pullthrough}
		a_kW(\fpn)\ppn = (-v(k)\Rnp k+\fpn(k))\ppn.
	\end{equation}
\end{lem}
\begin{proof}
	Proofs for the classical pull-through formula \cref{eq:standardpt} can for example be found in \cite{Dam.2018,Hinrichs.2022}. Combined with the transformation behavior of Weyl operators, cf. \cite[Lemma~D.14]{DamMoller.2018a}, \cref{eq:pullthrough} follows.
\end{proof}
The main result of this \lcnamecref{sec:conv.approx} now is the following statement, which treats both the infrared-regular case \subcref{thm:compactness.infraredreg} and the infrared-critical case \subcref{thm:compactness.infraredcrit}.
\begin{thm}\label{thm:compactness}
	Assume
	\begin{equation}\label{eq:compactassump}
		P\in\bigcap_{n\in\IN} \cI_0(\disp,\omega_n,v) \quad \mbox{and}\quad   \sup_{n\in\IN}\DPn < \infty.
	\end{equation}
	\begin{enumthm}
		\item\label{thm:compactness.infraredreg} If $\omega^{-1}v\in L^2(\IR^d)$, then the set \begin{equation}\label{eq:infraredregset}
			\bigcup_{n\in\IN}\big\{\psi\in \cD(H_n(P))\big| \|\psi\|\le 1, H_n(P)\psi=E_n(P)\psi \big\}
		\end{equation} is relatively compact.
		\item\label{thm:compactness.infraredcrit} Let $\Cw$ be defined as in \cref{def:Cw} and assume  that we have $\sup_{n\in\IN} \Cwn<\infty$, $\sup_{n\in\IN}\Cwn|\nabla E_n(P)|<1$ as well as $\sup\{|\partial_i^2 E_n(P)|:n\in\IN, i=1,\ldots,d\}<\infty $.
		Then, using the notation from \cref{def:fQ}, $\fpn \coloneqq f_{\nabla E_n(P),n}\in L^2(\IR^d)$ for all $n\in\IN$ and  the set \begin{equation}\label{eq:infraredcritset}
			\bigcup_{n\in\IN}\big\{W(\fpn)\psi \big| \psi\in\cD(H_n(P)), \|\psi\|\le 1, H_n(P)\psi=E_n(P)\psi \big\}\end{equation} is relatively compact.
	\end{enumthm}
\end{thm}
\begin{rem}
	The dressing transformation given by $\fpn$ applied to the ground states in \subcref{thm:compactness.infraredcrit} differs from the dressing transformation $f_{\nabla E(P),n}$ applied in \cref{mainthm2}. This is due to the fact that $\fpn$ correctly describes the leading infrared critical behavior of the resolvents $\Rnp{\cdot}$, as can be seen from \cref{lem:resapprox}. The error induced by altering the dressing transformation, however, tends to zero as $n\to\infty$, by \cref{lem:Enconv}. More details are given in the proof of \cref{mainthm2}.
\end{rem}
\begin{proof}[{Proof of \subcref{thm:compactness.infraredreg}}]
	To prove the statement, we verify the assumptions of \cref{comfock} with $D_n = B_{1/n}(0)$.
	
	Fix $P$ as in \cref{eq:compactassump} and throughout this proof write
	\[ \wt\Delta_P = \sup_{n\in\IN}\DPn.\] Further, assume that $\ppn$ is an arbitrary normalized ground state of $H_n(P)$ for any $n\in\IN$.
	By \cref{lem:apriori,lem:pullthrough} and the assumption, we immediately find
	\begin{equation*}
		\|a_k\ppn\| \le \DPn\frac{|v(k)|}{\omega_n(k)} \qquad \mbox{for almost all}\ k\in\IR^d \ \mbox{and all}\ n\in\IN.
	\end{equation*}
	Hence, by \cref{hyp:wn}, assumption \cref{comfock.L2} of \cref{comfock} is satisfied for the set \cref{eq:infraredregset} with $g = \wt\Delta_P\omega^{-1}v$.
	
	Now, fix $m\in\IN$. 
	By the resolvent identity and \cref{lem:pullthrough,lem:dispdiff}, for almost all $k,p\in\IR^d$, we have
	\begin{equation}\label{eq:resolventidentity}
		\begin{aligned}
			a_{k+p}\ppn  & - a_k\ppn 
			=   \left(v(k+p)-v(k)\right)\Rnp{k+p}\ppn \\
			& + v(k)\Rnp{k+p} \left(p \cdot \nabla \disp(P-k-\Pf) + D_{P-k}(p)\right)\Rnp{k}\ppn.
		\end{aligned}
	\end{equation}	
	Let
	\[ \alpha_m = \inf\{\omega(k):k\in B_{1/(m+1)}(0)^\sfc\} >0. \]
	and from now on assume $|p|<\frac 1m - \frac{1}{m+1}$.
	Using \cref{lem:apriori}, we directly find
	\begin{equation}\label{eq:IRreg-transupbound1}
		\|(v(k+p)-v(k)) \Rnp{k+p}\ppn\| \le \frac{\wt\Delta_P}{\alpha_m}|v(k+p)-v(k)|
	\end{equation}
	for all $k\in B_{1/m}(0)^\sfc$.
	Further, by \cref{lem:dispderbound} and the trivial estimate $\disp(P-k-\Pf) \le H_n(P-k)-E_n(P)$, we find
	\[ \||\nabla \disp(P-k-\Pf)|\Rnp{k}\| \le \CT{1,1}\|\Rnp k\| + \CT{1,2}. \]
	Combining these observations and again using \cref{lem:apriori} yields
	\begin{equation}\label{eq:IRreg-transupbound2}
		\begin{aligned}
			&\|v(k)\Rnp{k+p} \left(p \cdot \nabla \disp(P-k-\Pf) + D_{P-k}(p)\right)\Rnp{k}\ppn\|
			\\ & \qquad \qquad 
			\le
			\frac{\wt\Delta_P}{\alpha_m}|p| \left((\CT{1,1}+|p|\CT2)\frac{\wt\Delta_P}{\alpha_m} + \CT{1,2}\right)|v(k)|
		\end{aligned}
	\end{equation}
	for all $k \in B_{1/m}(0)^\sfc$.
	Combining \cref{eq:resolventidentity,,eq:IRreg-transupbound1,,eq:IRreg-transupbound2}, we have two $n$-independent constants $C_1,C_2>0$ such that
	\[
	\left(\int_{B_{1/m}(0)^\sfc} \big\|a_{k+p}\ppn   - a_k\ppn\big\|^2\d k\right)^{1/2}
	\le C_1\|v(\cdot+ p) - v(\cdot)\|_2 + |p|C_2\|v\|_2.
	\]
	Since $v\in L^2(\IR^d)$,
	the right hand side vanishes as $p\to 0$, so assumption \cref{comfock.cont} of \cref{comfock} is satisfied for the set \cref{eq:infraredregset}. Hence, the statement follows from \cref{comfock}.
\end{proof}
\begin{proof}[{Proof of \subcref{thm:compactness.infraredcrit}}]
	First, we note that $\fpn\in L^2(\IR^d)$ follows from the assumptions similar to the argument in \cref{rem:dresstrans}.
	
	The proof now
	is very similar to the above one.
	Again, we verify the assumptions of \cref{comfock} with $D_n = B_{1/n}(0)$.
	
	Fix $P$ as in \cref{eq:compactassump}. Further, assume that $\ppn$ is an arbitrary normalized ground state of $H_n(P)$ for any $n\in\IN$.
	By the assumptions and \cref{lem:resapprox,lem:pullthrough}, there exists an $n$-independent constant $C>0$ such that
	\[
	\|a_k W(\fpn)\ppn\| \le C(1+\omega_n^{-1/2}(k))|v(k)| \quad \mbox{for almost all}\ k\in\IR^d\ \mbox{and all}\ n\in\IN.
	\]
	By \cref{hyp:wn}, this verifies assumption \cref{comfock.L2} of \cref{comfock} for the set \cref{eq:infraredcritset} with $g = C(1+\omega^{-1/2})v$.
	
	From now, we assume $|p|<\frac 1m-\frac 1{m+1}$.
	Following the lines of \cref{eq:resolventidentity,eq:IRreg-transupbound1,eq:IRreg-transupbound2}, we find
	\begin{align*}
		&\left(\int_{B_{1/m}(0)^\sfc} \big\|a_{k+p}W(\fpn)\ppn   - a_kW(\fpn)\ppn\big\|^2\d k\right)^{1/2} \\
		&\qquad \le C_1\|v(\cdot+ p) - v(\cdot)\|_2 + |p|C_2\|v\|_2 +  \|\chr{B_{1/m}(0)^\sfc}(\fpn(\cdot + p) - \fpn(\cdot))\|_2.
	\end{align*}
	By definition and using a similar argument as in the derivation of \cref{eq:IRreg-transupbound1}, we see
	\begin{align*}
		\|\chr{B_{1/m}(0)^\sfc} & (\fpn(\cdot + p) - \fpn(\cdot))\|_2\\
		\le &  \frac{1}{\alpha_m(1-\Cwn |\nabla E_n(P)|)}\|v(\cdot + p)-v(\cdot)\|_2
		\\ &  + \frac{ |p\cdot \nabla E_n(P)| }{\alpha_m^2(1-\Cwn |\nabla E_n(P)|)^2}\left\|v\right\|_2
		\\& + 
		\left\| \frac{|\omega_n(\cdot + p)- \omega_n|}{(\omega_n(k)-k\cdot \nabla E_n(P))(\omega_n(k+p)-(k+p)\cdot \nabla E_n(P))}  v\right\|_2
		.
	\end{align*}
	Since $(\alpha_m(1-\Cwn|\nabla E_n(P)|))^{-1}$ is uniformly bounded in $n\in\IN$ and since $v\in L^2(\IR^d)$, the first two lines tend to zero as $p\to 0$ uniformly in $n\in\IN$. Further, combining \cref{eq:omegadiff} and the dominated convergence theorem, the same accounts for the last line.
	Hence,
	assumption \cref{comfock.cont} of \cref{comfock} is satisfied for the set \cref{eq:infraredcritset} and the statement follows.
\end{proof}
In the infrared-regular case, \cref{thm:compactness} entails the following simple existence result, which is independent of the concrete choices of $\disp$, $\omega$ or $v$.
\begin{cor}\label{cor:irregulargs}
	If $\omega^{-1}v\in L^2(\IR^d)$ and $P\in\IR^d$ satisfies \cref{eq:compactassump}, then $E(P)$ is a simple eigenvalue of $H(P)$ for all $P\in\cI\coloneqq \bigcap_{n\in\IN}\cI_0(\disp,\omega_n,v)$.
	In this case, if $\ppn\in\pos$ denotes the unique ground state of $H_n(P)$, then the eigenspace of $H(P)$ corresponding to $E(P)$ is spanned by $\psi_{P,\infty}\coloneqq\lim_{n\to\infty}\ppn\in\pos$.
	Further, there exists a zero set $\cN\subset\cI$ such that $\cI\setminus\cN\to\FS$, $P\mapsto \psi_{P,\infty}$ is continuous.
\end{cor}
\begin{rem}
	Employing a variety of techniques, existence of ground states in the infrared-regular case is a well-known result in the literature, see for example \cite{Frohlich.1973,BachFroehlichSigal.1998a,Spohn.1998,Gerard.2000,LorincziMinlosSpohn.2002,GubinelliHiroshimaLorinczi.2014}. Since the statement follows very easily from our compactness theorem, we give the simple proof nevertheless.
	Further, to the authors knowledge, the observation that continuity out of a zero set follows directly by combining our compactness theorem with positivity argument did not previously appear in the literature.
\end{rem}
\begin{proof}
	Let $\ppn\in \pos$ be the unique positive ground state of $H_n(P)$, cf. \cref{prop:Moller,prop:uniqueness}. By \cref{thm:compactness.infraredreg}, there exists a subsequence $(\ppnk)_{k\in\IN}$ which converges to a normalized vector $\ppi$. It immediately follows that $\ppi\in\pos\cap \cD(H(P))$, since $\pos$ is closed and since $\|H(P)\ppnk\|$ is uniformly bounded in $k\in\IN$. By the lower-semicontinuity of closed quadratic forms, cf. \cite[\textsection VI, Theorem~1.16]{Kato.1980}, \cref{lem:minimizingsequence} further implies 
	\[  0 \le \braket{\ppi,(H(P)-E(P))\ppi} \le \liminf_{k\to\infty}\braket{\ppnk,(H(P)-E(P))\ppnk} = 0, \]
	so $\ppi$ is in fact a ground state of $H(P)$. Since subsets of relatively compact subsets are relatively compact, we can repeat this argument for any subsequence of $(\ppn)$. By \cref{prop:uniqueness}, the strong limit will always be the same, which entails $\ppi = \lim_{n\to\infty}\ppn$.
	
	It remains to prove the continuity statement. To this end, we first observe that we can choose normalized 	$\wt\psi_{P,n}\in \ker (H(P)-E(P))$ for all $P\in\cI$, $n\in\IN$ such that $\cI\to\FS$, $P\mapsto \wt\psi_{P,n}$ is continuous, e.g., by the analyticity used in \cref{lem:perturb}.
	Now, we observe that the map
	$\fp:\FS\to \pos$  given by
	\[ (\fp\Phi)_0=\abs{\Phi_0},\quad  (\fp\Phi)_\ell (k_1,\ldots,k_\ell) \coloneqq (-1)^\ell\abs{\Phi_\ell(k_1,\ldots,k_\ell)}\wt v(k_1)\cdots \wt v(k_\ell), \]
	where $\wt v(k_i)=0$ if $v(k_i)=0$ and $\wt v(k_i)=v(k_i)/\abs{v(k_i)}$ else,
	is continuous, by the continuity of the absolute value and the dominated convergence theorem.
	Since by \cref{prop:uniqueness} $\fp\wt\psi_{P,n} = \ppn$, we find that $P\mapsto \ppn$ is continuous on $\cI$ for any $n\in\IN$. Further, since it converges pointwise, Egorov's theorem implies that for all $R>0$ and any set $\eps>0$ there exists a set $D_\eps\subset \cI\cap B_R(0)$ of Lebesgue measure smaller than $\eps$ such that the convergence $\ppn \xrightarrow{n\to\infty} \psi_{P,\infty}$ is uniform in $P\in \cI\cap B_R(0)\setminus D_\eps$, which in turn implies that $P\mapsto \psi_{P,\infty}$ is continuous on $\cI\cap B_R(0)\setminus D_\eps$. The statement now follows taking the limit $\eps\downarrow 0$ and choosing $R$ sufficiently large.
\end{proof}
\subsection{Mass Shell Derivative Bounds}\label{sec.conv:ex}

In this final \lcnamecref{sec.conv:ex}, we study the behavior of the mass shell $P\mapsto E(P)$. This allows us to verify the assumptions of \cref{thm:compactness} for a region of total momenta $P$ depending only on the constants from our inital Hypotheses \cref{hyp:Omega,hyp:div,hyp:omega,hyp:v,def:Cw}.
\begin{thm}\label{thm:massshell}We have the following properties of the mass shell.
	\begin{enumthm}
		\item\label{thm:massshell.td} $E(\cdot)$ is almost everywhere twice differentiable.
		\item\label{thm:massshell.dc} If $E(\cdot)$ is differentiable at $P$, then $\nabla E(P) = \lim_{n\to\infty} \nabla E_n(P)$.
		\item\label{thm:massshell.tds} If $E(\cdot)$ is twice differentiable at $P$, then $\limsup_{n\to\infty} \max_{i=1,\ldots,d} \partial_i^2E_n(P)<\infty$.
		\item\label{lem:massshell.convexlowbound} $E(P-k)-E(P) \ge \begin{cases} -2\CT2 |k||P| & \mbox{for}\ |k|\le |P|,\\-\CT2|P|^2 & \mbox{else}. \end{cases} $
		\item\label{lem:massshell.dbound} If $E(\cdot)$ is differentiable at $P$, then $|\nabla E(P)| \le \CT2|P|$.
	\end{enumthm}
\end{thm}
\begin{proof}
	Let us first recall the well-known fact that
	\begin{equation}\label{eq:massshellbounds}
		0\le E(P)-E(0) \le  \CT2 |P|^2 \qquad \mbox{for all}\ P\in\IR^d.
	\end{equation}
	The first inequality was proven by Gross \cite{Gross.1972} and easily extends to our generalized case; also see \cite[Lemma 4.2]{Dam.2018}. The second inequality can be verified (at least for the massive case) by using the variational argument from \cite[p.~207]{Spohn.2004} in conjunction with \cref{lem:dispdiff}.
	Further, for fixed $\psi\in\cD(H(0))$, let us consider the map $\xi \mapsto f_\psi(\xi) = \Braket{\psi, \big(  \CT 2 \lvert\xi\lvert^2 - H(\xi)\big)\psi}$.
	It is analytic and the Hessian is easily calculated to be
	\[ \sfH f_\psi(\xi) = 2\CT 2 \nn{\psi}^2 \Id - \braket{\psi, \sfH\disp(\xi-\Pf)\psi} \]
	and is hence, by \cref{hyp:Omega}, positive definitive.
	Thus, the function $f_\psi$ is convex.
	Since the pointwise supremum of convex functions is convex, we find that
	\begin{equation}\label{eq:massshellconvex}
		P\mapsto \CT 2 |P|^2 - E(P) { = \sup_{\substack{\psi \in \cD(H(P))\\\nn{\psi}=1}}  f_\psi(P)  } \qquad\mbox{is convex}.
	\end{equation}
	Hence, \subcref{thm:massshell.td} follows directly from the {Alexandrov theorem} \cite{Alexandrov.1939}, cf. \cref{alex}.
	Further, \cref{eq:massshellconvex,lem:convderivative} prove \subcref{thm:massshell.dc} and \subcref{thm:massshell.tds}.
	Applying \cref{lem:convexdiff} and noting that we already verified the assumptions therein in \cref{eq:massshellbounds,eq:massshellconvex} shows \subcref{lem:massshell.convexlowbound}.
	Finally, since $P=0$ is a global minimum of $E$ by \cref{eq:massshellbounds} and since the derivative of convex functions is increasing along any line, \subcref{lem:massshell.dbound} follows from \cref{eq:massshellconvex}.
\end{proof}
\begin{cor}\label{cor:gapnr}
	Assume $\mw>0$ and $\Cw<\infty$. Then for all $P\in\IR^d$ with $|P|<(2\CT2\Cw)^{-1} $, we have $P\in \cI_0(\disp,\omega,v)$ and $\DP \le ({1-2\Cw\CT2|P|})^{-1} $.
\end{cor}
\begin{proof}
	Dividing \cref{lem:massshell.convexlowbound} by $\omega(k)$, we obtain
	\[  \frac{E(P-k)-E(P)}{\omega(k)}\ge -2\CT2 \Cw |P|. \]
	Hence, the statement immediately follows from \cref{lem:I0char} and the definition \cref{def:DP}.
\end{proof}
For the non-relativistic dispersion relation, the above two statements already suffice to check the conditions, see the proof of \cref{mainthm2}. For the relativistic dispersion relation, we extend them in the following statement.
\begin{lem}\label{lem:massshellsr}
	Assume $|\nabla \disp(p)|\le \CT{1,1}$ for all $p\in\IR^d$.
	\begin{enumlem}
		\item For all $P\in\IR^d$, there exists $\CP\in[0,\CT{1,1}]$ such that $E(P-k)-E(P)\ge -\CP|k|$ for all $k\in\IR^d$. If $|\nabla \disp(p)|<\CT{1,1}$ for all $p\in\IR^d$, then we can choose $\CP<\CT{1,1}$ for all $P\in\IR^d$.
		\item Assume $\mw>0$ and assume $\CP\Cw< 1$ for some $P\in\IR^d$. Then $P\in \cI_0(\disp,\omega,v)$ and $\DP \le 1/(1-\CP\Cw)$. Especially, if $\CT{1,1}\Cw<1$ or $\CT{1,1}\Cw=1$ and $|\nabla \disp(p)|<\CT{1,1}$ for all $p\in\IR^d$, then $\cI_0(\disp,\omega,v)=\IR^d$.
	\end{enumlem}
\end{lem}
%
\begin{proof}
	Since
	\[ \braket{\psi,H(P)\psi} \ge \braket{\psi,\disp(P-\Pf)\psi} + (1-\eps)\braket{\psi,\Hf\psi}-C_\eps\|\psi\|^2, \]
	\cref{hyp:div} implies that there exists $C_{P,\delta}>0$ for all $\delta>0$ and $P\in\IR^d$ such that $\braket{\psi,|\Pf|\psi} \le C_{P,\delta}$ for all $\psi\in\cD(H(P))$ with $\braket{\psi,H(P)\psi}\le E(P)+\delta$.
	Hence, for any fixed $\delta>0$, set
	\[ \CP = \sup_{\substack{p\in\IR^d\\|p|\le C_{P,\delta}}}|\nabla \disp(p)|. \]
	\Cref{lem:diffeps} then directly implies that
	\[ |\nabla E(P)|\le \CP \qquad \mbox{at all points $P$ where $E$ is differentiable.}\]
	The remaining statements easily follow analogous to the proof of \cref{cor:gapnr}.
\end{proof}

\appendix

\section{Convex Functions and Derivatives}\label{app:regularity}

In this \lcnamecref{app:regularity}, we collect some simple statements about the derivatives of (in most cases) convex functions. These well-known facts are the main ingredients of the arguments collected in \cref{sec.conv:ex}.

\subsection{First Derivative Bounds}

The following statement allows us to prove existence of non-Fock ground states for almost all total momenta in the semi-relativistic case, cf. \cref{lem:massshellsr} and the proof of \cref{mainthm2}. Although its proof is near trivial, we couldn't find a similar statement in the literature.
\begin{lem}\label{lem:diffeps}
	Let $I$ be an open interval and let $\fF$ be a family of differentiable functions $f:I\to\IR$. Further, assume that $g:I\to \IR$ given by $g(x)=\inf\{f(x):f\in\fF\}$ is right-differentiable at $x=a$. Then
	\begin{equation}\label{eq:difflesup}
		g'_+(a) \le \lim_{\eps\downarrow 0}\sup\{f'(x):f\in\fF,x\in[a,a+\eps)\}.
	\end{equation}
\end{lem}
\begin{rem}\ 
	\begin{enumrem}
		\item The limit $\eps\downarrow0$ on the right hand side of \cref{eq:difflesup} cannot be removed, i.e., replaced by the simpler expression $\sup\{f'(a):f\in\fF\}$. A simple counterexample is the family $\fF = \{f_\delta:\IR\to\IR|\delta>0\}$ with
		\begin{align*}
			f_\delta(x) = \begin{cases} 2\delta & \mbox{for}\ x\le-\delta,\\ -\frac 1\delta(x+\delta)^2 + 2\delta & \mbox{for}\ x\in(-\delta,0), \\ \frac 1 \delta(x-\delta)^2 & \mbox{for}\ x\in[0,\delta), \\ 0 &\mbox{for}\ x\ge\delta. \end{cases}
		\end{align*}
		In this case $f_\delta'(0)= -2$ for all $\delta>0$, but $g(x)=\inf\{f_\delta(x):\delta>0\}=0$ for all $x\in\IR$, so $g'(0)>\sup\{f_\delta'(0):\delta>0\}$. This, however, does not contradict \cref{eq:difflesup}, because $f_\delta'(\pm\delta) = 0$ for all $\delta>0$.
		\item A similar statement holds for left-differentiable $g$. Explicitly, in this case
		\begin{align*}
			g'_-(a) \ge \lim_{\eps\downarrow 0} \inf \{ f'(x):f\in\fF,x\in(a-\eps,a]\}.
		\end{align*}
	\end{enumrem}
\end{rem}
\begin{proof}
	We fix $\eps>0$ such that $a+\eps\in I$, an arbitrary zero sequence $(h_n)\subset (0,\eps)$ and a sequence $(f_n)\subset \fF$ such that $0\le f_n(a)-g(a)\le h_n^2$.
	Then, by the definition of $g$ and the mean value theorem, we find
	\begin{align*}
		\frac{g(a+h_n)-g(a)}{h_n} &\le \frac{f_n(a+h_n)-f_n(a)}{h_n} + \frac{f_n(a)-g(a)}{h_n}\\
		& \le \sup\{f'(x):f\in\fF,x\in[a,a+\eps)\} + h_n.
	\end{align*}
	The statement now follows, by first taking the limit $n\to\infty$ and afterwards the limit $\eps \downarrow 0$, where the latter exists due to monotonicity.
\end{proof}

\subsection{Convexity Properties}

Bounds obtained by convexity of the function $P\mapsto C P^2- E(P)$, where $E(\cdot)$ is the mass shell of a Nelson-type model and $C>0$ is an appropriate constant, have been used throughout the literature, see for example \cite{LossMiyaoSpohn.2007,KoenenbergMatte.2014,HaslerSiebert.2022}. We collect the essential statements for our proofs below.

The first statement is the well-known Alexandrov's theorem.
\begin{lem}[{\cite{Alexandrov.1939}}]\label{alex}
	If $f:\IR^d\to\IR$ is convex, then $f$ is almost everywhere twice differentiable.
\end{lem}
The next two statement concerns the convergence of first and second derivatives of convex functions.
\begin{lem}[{\cite[Lemmas~C.6~\&~C.7]{HaslerSiebert.2022}}]
	\label{lem:convderivative}
	Let $(f_n)_{n\in\IN},f$ be convex functions $\IR\to\IR$ and assume that $f_n\xrightarrow{n\to\infty} f$ pointwise.
	\begin{enumlem}
		\item For all $x\in \IR$, where the derivatives $f_n'(x)$ and $f'(x)$ exist, we have $f_n'(x)\xrightarrow{n\to\infty}f'(x)$.
		\item For all $x\in \IR$, where both $f_n''(x)$ and $f''(x)$ exist, we have $\liminf_{n\to\infty}f_n''(x)<\infty$.
	\end{enumlem} 
\end{lem}
The following is a generalization of \cite[Appendix A]{LossMiyaoSpohn.2007}.
Given a convex function $f:\IR\to[0,\infty)$, we denote the set of non-negative convex functions majorized by $f$ as
\[ \cC_f = \{g:\IR\to \IR \mbox{ convex}\colon 0\le g \le f\}.  \]
Further, we define
\[ \Delta_f(p,q) = \sup_{g\in \cC_f} g(p)-g(q) \qquad\mbox{for all}\ p,q\in\IR^n.  \]
\begin{lem}\label{lem:convex}
	Fix $c > 0$, $b\in\IR$ and $a\ge b^2/2c^2$. Let $p:\IR\to\IR$ be the quadratic function $p(x)=\frac c2x^2 + bx + a$. Then
	\begin{equation}\label{eq:parabola}
		\Delta_p(0,1) = \begin{cases} \sqrt{2ac} + b & \mbox{if}\ 2a \ge c, \\ p(1) & \mbox{else}.  \end{cases}
	\end{equation}
\end{lem}
\begin{proof}
	First, we observe that the function $h:\IR\to \IR$ with
	\[  h(x) = \begin{cases} 0 & \mbox{if}\ x\le 0, \\ (\sqrt{2ac} + b) x & \mbox{if}\ a> 0,\ 0 < x\le 2a/c, \\ p(x) & \mbox{else}, \end{cases} \]
	is convex and majorized by $p$, so it suffices to prove that $g(1)-g(0)$ is smaller than the upper bound in \cref{eq:parabola}.
	
	Assume $y\in[0,a]$. We want to analyze all $g\in\cC_p$ with $g(0)=y$. To that end, we first calculate the possible positive point of contact $x_0\ge 0$  with any tangent to the graph of $p$ through $(0,y)$. This means solving the quadratic equation
	\[ y + x_0p'(x_0) = p(x_0),\ x_0\ge 0 \iff x_0 = \sqrt{\frac2c(a-y)}.  \]
	Now all $g\in\cC_p$ with $g(0)=y$ lie below the line through $(0,y)$ and $(x_0,p(x_0))$ on the interval $[0,x_0]$. This yields a bound on $g(1)-g(0)$ if $x_0> 1$, i.e., $2(a-y)\ge c$. For all $g\in\cC_p$ with $g(0)=y$, we obtain
	\begin{equation}\label{eq:upperboundy}
		g(1)-g(0) \le \begin{cases}  \sqrt{2c(a-y)}+b & \mbox{if}\ 2(a-y)> c,\\ \frac c2+b+a - y & \mbox{if}\ 2(a-y)\le c.	\end{cases}
	\end{equation}
	It now suffices to prove
	\begin{equation*}
		g(1)-g(0) \le \begin{cases} \sqrt{2ac}+b & \mbox{if}\ 2a> c,\\ \frac c2+b+a & \mbox{if}\ 2a\le c,	\end{cases}
	\end{equation*}
	i.e., that the upper bound in \cref{eq:upperboundy} is maximized if $y=0$. This immediately follows from \cref{eq:upperboundy} if $2(a-y)>c$ and $2a\le c$, so it remains to treat the case $c<2a\le c+2y$ or equivalently $c-2y< 2(a-y) \le c$. However, in this case
	\[ \frac c2 + a -y \le c = \sqrt{c\cdot c} < \sqrt{2ac},  \]
	which finishes the proof.
\end{proof}
As a simple corollary, we obtain the following.
\begin{cor}\label{lem:convexdiff}
	Let $C>0$ and
	let $F:\IR^d\to\IR$ satisfy
	\begin{enumerate}[label = {\upshape (\roman*)}, ref = {\upshape \roman*}]
		\item $F(0)\le F(P)$ for all $P\in\IR^d$,
		\item $F(P)\le \frac C2 |P|^2 + F(0)$ for all $P\in\IR^d$,
		\item $P\mapsto \frac C2|P|^2-F(P)$ is convex.
	\end{enumerate}
	Then
	\[ F(P-k)-F(P) \ge \begin{cases} -C|k||P|+\frac C2|k|^2 & \mbox{if}\ |k|\le |P|,\\ -\frac C2|P|^2 & \mbox{else}.\end{cases} \]
\end{cor}
\begin{proof}
	We apply \cref{lem:convex} with the choice $p(t) = \frac C2 |P-tk|^2$, i.e., $c = C|k|^2$, $b=-C P\cdot k$ and $a=\frac C2|P|^2$.
	Hence, since the function $g(t) = p(t) - F(P-tk) + F(0)$ is convex and satisfies $0\le g\le p$, we find
	\begin{align*}
		\frac C2|P-k|^2 - F(P-k)  & - \frac{C}{2}|P|^2 + F(P) \\&= \frac C2|k|^2 - \frac CP\cdot k - \big(F(P-k)-F(P)\big) \\
		& \le \begin{cases}  C|k||P| - C P\cdot k & \mbox{if}\ |k|\le|P|, \\ \frac C2|P-k|^2 & \mbox{else}.		\end{cases}
	\end{align*}
	Simple rearrangement yields the statement.
\end{proof}

\section{Simple Properties of Fock Space Operators}\label{sec:operator}

In this \lcnamecref{sec:operator}, we collect well-known properties of Fock space operators.

\subsubsection*{Dispersion Relation Derivative Bounds}

\begin{lem}Assume that \cref{hyp:Omega} holds. 
	\begin{enumlem}
		\item\label{lem:dispindP} $\cD(\disp(P-\Pf))$ is independent of $P$.
		\item\label{lem:dispderbound} $|\nabla \disp (P-\Pf) | \le \CT{1,1} + \CT{1,2} \disp(P-\Pf)$.
		\item\label{lem:dispdiff} For all $P,k\in\IR^d$ there exists $D_P(k)\in \cB(\FS)$ with $\|D_P(k)\| \le  \CT{2} |k|^2 $ such that
		\[ \disp(P-k-\Pf)-\disp(P-\Pf) = k\cdot \nabla \disp (P-\Pf) + D_P(k) \quad\mbox{on}\ \cD(\disp(P-\Pf)).  \]
	\end{enumlem}
\end{lem}
\begin{proof}
	Details of the straightforward proofs can be found in \cite[Lemma 3.2]{Dam.2018}
\end{proof}

\subsubsection*{Field Operator Commutators and Relative Bounds}

\begin{lem}[{\cite[Prop.~5.12]{Arai.2018}}]Assume that $A$ is a non-negative and injective selfadjoint multiplication operator on $L^2(\IR^d)$.
	\label{th:phi estimate}
	If $f \in \cD(A^{-1/2})$, then $\cD(\ph(f))\supset\cD(\dG(A)^{1/2})$ and for all $\psi\in\cD(\dG(A)^{1/2})$
	\begin{align*}
		\nn{\ph(f)\psi} \leq 2\left(\nn{f} + \nn{A^{-1/2} f}\right) \nn{ (\dG(A) + 1)^{1/2} \psi}.
	\end{align*}
\end{lem}
\begin{lem}[{\cite[Theorem~5.17]{Arai.2018}}]
	Assume that $A$ is a non-negative and injective selfadjoint multiplication operator on $L^2(\IR^d)$.
	If $f\in \cD(A)\cap\cD(A^{-1/2})$, then $\cD(\ph(f))\cap \cD(\ph(\i A f))\supset \cD(\dG(A)^{3/2})$, $\ph(f)\cD(\dG(A)^{3/2}) \subset \cD(\dG(A))$ and
	\[  [\dG(A),\ph(f)] = \i\ph(\i A f) \qquad\mbox{holds on}\ \cD(\dG(A)^{3/2}). \]
\end{lem}

\subsubsection*{Weyl Operator Transformation Properties}

\begin{lem}[{\cite[Prop.~5.2.4]{BratteliRobinson.1996}}] \label{lem:transform a}  Assume $f , g \in L^2(\IR^d)$.
	\[
	W(f) \cD(\ph(g)) = \cD(\ph(g)) \quad \mbox{and}\quad W(f) \ph(g)  W(f)^* = \ph(g) - 2\Re \sc{  f , g } .
	\]
\end{lem}

\begin{lem}[{\cite[Lemma~B.3]{DamHinrichs.2021}}] \label{lem:transform dG}    Let $A$ be a selfadjoint operator on $L^2(\IR^d)$ and let $f \in \cD(A)$. Then $W(f) \cD(\dG(A)) = \cD(\dG(A)) \subset \cD(\ph(Af)) $ and   
	\begin{equation} \label{eq:wdgammarel} 
		W(f)  \dG(A)  W(f)^* = \dG(A) - \ph(  A f ) +  \sc{ f , A f }  .
	\end{equation} 
\end{lem} 

\bibliographystyle{halpha-abbrv}
\bibliography{lit}

\end{document}